\documentclass{amsart}
\usepackage[dvips]{epsfig}
\usepackage{graphicx}
\usepackage{latexsym}
\usepackage{amsmath}
\usepackage{amsthm}
\usepackage{amssymb}
\RequirePackage[colorlinks=true,citecolor=blue,urlcolor=blue]{hyperref}

\usepackage{xcolor}

 \usepackage[applemac]{inputenc}

\newtheorem{thm}{Theorem}
\newtheorem{assumption}[thm]{Assumption}
\newtheorem{lem}[thm]{Lemma}

\newtheorem{defi}[thm]{Definition}

\newtheorem{prop}[thm]{Proposition}

\newcommand{\vip}{\vskip.2cm}
\newcommand{\R}{{\mathbb{R}}}

\newcommand{\E}{\mathbb{E}}
\newcommand{\COMMENTAIRE}[1]{}
\newcommand{\PP}{{\mathbb{P}}}

\begin{document}

\title[Nonparametric estimation in age dependent branching processes]{Nonparametric estimation of the division rate of an age dependent branching process}

\author{Marc Hoffmann and Ad\'ela\"ide Olivier}

\address{Marc Hoffmann, CEREMADE, CNRS-UMR 7534,
Universit\'e Paris-Dauphine, Place du mar\'echal De Lattre de Tassigny
75775 Paris Cedex 16, France.}

\email{hoffmann@ceremade.dauphine.fr}

\address{Ad\'ela\"ide Olivier, CEREMADE, CNRS-UMR 7534, Universit\'e Paris-Dauphine, Place du mar\'echal De Lattre de Tassigny
75775 Paris Cedex 16, France.}

\email{olivier@ceremade.dauphine.fr}

\begin{abstract} We study the nonparametric estimation of the branching rate $B(x)$ of a supercritical Bellman-Harris population: a particle with age $x$ has a random lifetime governed by $B(x)$; at its death time, it gives rise to $k \geq 2$ children with lifetimes governed by the same division rate and so on. We observe in continuous time the process over $[0,T]$. Asymptotics are taken as $T \rightarrow \infty$; the data are stochastically dependent and one has to face simultaneously censoring, bias selection and non-ancillarity of the number of observations.
In this setting, under appropriate ergodicity properties, we construct a kernel-based estimator of $B(x)$ that achieves the rate of convergence $\exp(-\lambda_B \frac{\beta}{2\beta+1}T)$, where $\lambda_B$ is the Malthus parameter and $\beta >0$ is the smoothness of the function $B(x)$ in a vicinity of $x$. We prove that this rate is optimal in a minimax sense and we relate it explicitly to classical nonparametric models such as density estimation observed on an appropriate (parameter dependent) scale. We also shed some light on the fact that estimation with  kernel estimators based on data alive at time $T$ only is not sufficient to obtain optimal rates of convergence, a phenomenon which is specific to nonparametric estimation and that has been observed in other related growth-fragmentation models.  
\end{abstract}

\maketitle



\section{\textsc{Introduction}}

\subsection{Motivation}

Structured models have been paid particular attention over the last few years, both from a probabilistic and an applied analysis angle, in particular with a view toward a better understanding of population evolution in mathematical biology (see for instance the textbook by Perthame~\cite{Perthame} and the references therein). In this context, a more specific focus and need for statistical methods has emerged recently ({\it e.g.} Doumic {\it et al.}~\cite{DPZ, DHRR, DHKR1} and the references therein) and this is the topic of the present paper. If $x$ denotes a so-called structuring variable -- for instance age, size, any measure of variability or DNA content of a cell or bacteria,  and if $n(t,x)$ denotes the number or density of cells at time $t$ of a population starting from a single ancestor at time $t=0$, a sound mathematical model can be obtained by specifying an evolution equation for $n(t,x)$. 

\vip

Consider for instance the paradigmatic problem of age-dependent cell division, where the evolution of $n(t,x)$ is given by the simplest transport-fragmentation equation 
\begin{equation} \label{det description}
\left\{
\begin{array}{l}
\displaystyle \frac{\partial}{\partial t}n(t,x) + \frac{\partial }{\partial x} n(t,x) + B(x)n(t,x) = 0 \\ \\
n(t,0) = m\int_0^\infty B(y)n(t,y)dy,\;\;t>0,\;\;n(0,x) = \delta_0,
\end{array}
\right.
\end{equation}
where $\delta_0$ denotes the Dirac mass at point $0$. In this model, each cell dies according to a division rate $x \leadsto B(x)$ that depends on its age $x$ only (a living cell of age $x$ has probability $B(x)dx$ of dying in the interval $[x, x+dx]$) and, at its time of death, it gives rise to $m \geq 2$ children at its time of death. The parameters $(m,B)$ specify the so-called age-dependent model. 

\vip

In this seemingly simple context, we wish to draw statistical inference on the division rate function $x \leadsto B(x)$ and on $m$ in the most rigorous way, when we observe the evolution of the population through time and when the shape of the function $B$ can be arbitrary, to within a prescribed smoothness class, {\it i.e.} in a nonparametric setting. In order to do so, we transfer the deterministic description \eqref{det description} into a probabilist model that consists of a system of (non-interacting) particles specified by a probability distribution $p$ on the integers  (the offspring distribution) and  a probability density $f$ on $[0,\infty)$.
A particle has a random lifetime drawn according to $f(x)dx$; at the time of its death, it gives rise to $k$ children with probability $p_k$ (with $p_0=p_1=0$), each child having independent lifetimes distributed as $f(x)dx$, and so on. The resulting process  is a classical supercritical Bellman-Harris, see for instance the textbooks of Harris~\cite{Harris} or Athreya and Ney~\cite{AthreyaNey}. It is described by a piecewise deterministic Markov process 
\begin{equation} \label{eq:def X}
X(t)=\big(X_1(t), X_2(t),\ldots\big), \quad t \geq 0,
\end{equation}
with values in $\bigcup_{k \geq 1}[0,\infty)^k$, where the $X_i(t)$'s denote the (ordered) ages of the living particles at time $t$. The formal link between $X(t)$ and $n(t,x)$ is obtained via $n(t,x) = \E\big[\sum_{i = 1}^\infty \delta_{X_i(t)=x}\big]$ which has to be understood in a weak (measure) sense, {\it i.e.} the empirical measure (in expectation) of the particle system and solves Equation \eqref{det description}, we refer to \cite{Oelschlager}.\\

The correspondence between $(m,B)$ and $(f,p)$ is given
by 
\begin{equation} \label{hazard formula}
B(x)=\frac{f(x)}{1-\int_0^x f(s)ds},\;x\in [0,\infty),\;\;\text{and}\;\;m = \sum_{k \geq 2}kp_k,
\end{equation}
provided everything is well defined. 
Under fairly reasonable assumptions described below, it is one-to-one between $B$ and $f$, but not between $m$ and $p$.
We are interested in the nonparametric estimation of $x\leadsto B(x)$, which is nothing but the hazard rate function of the lifetime density $f$ of each particle, and also in the mean offspring $m$, the whole distribution $p$ being considered as a nuisance parameter. 

\subsection{Objectives and results}

\begin{figure}[h!]
\centering
\includegraphics[width=6cm]{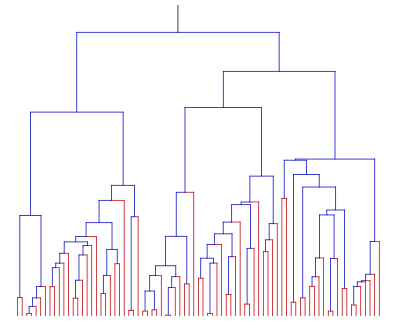}
\includegraphics[width=6cm]{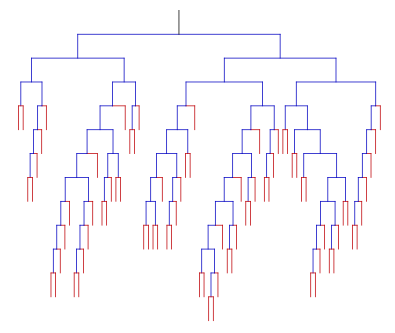}
\caption{{\it The effect of bias selection. Simulation of a binary ($p_2=1$ so $m=2$) age-dependent tree with $B$ given in Section~\ref{section: simus}, up to time $T=8$ ($|\mathcal T_T| = 145$). Left: the size of each segment represents the lifetime of an individual. Individuals alive at time $T$ are represented in red. Right: genealogical representation of the same realisation of the tree.} \label{Fig 1}}
\end{figure}

\subsubsection*{Observation schemes}

We assume we observe the whole trajectory $(X(t), t \in [0,T])$, where $T>0$ is a fixed (large) terminal time. Asymptotics are taken as $T \rightarrow \infty$. If we denote by 
$\mathcal T_T$ the population of individuals that are born before $T$ and observed up to time $T$ and if $(\zeta_u^T, u \in \mathcal T_T)$ denotes the values of the ages of the different individuals of $\mathcal T_T$  (at their time of death or at time $T$), we wish to draw inference on $B(x)$ based on 
\begin{equation*} 
\big\{X(t), t \in [0,T]\big\} = \big\{\zeta_u^T, u \in \mathcal T_T\big\}.
\end{equation*}

Although the lifetimes of the individuals are independent (and identically distributed)  with common density $f$, this is no longer the case for the population $(\zeta_u^T, u \in \mathcal T_T)$ considered as a whole: the tree structure plays a crucial role and we have to face several non-trivial difficulties:
\begin{itemize}
\item[{\bf 1)}] {\it Bias selection:} particles with small lifetimes are more often observed than particles with large lifetimes since the observation of the process is stopped along all the branches at the fixed time $T$, as illustrated in Figure~\ref{Fig 1}.
\item[{\bf 2)}] {\it Censoring:} if $\partial \mathcal T_T \subset \mathcal T_T$ denotes the population of individuals alive at time $T$ (in red in Figure~\ref{Fig 1}), they are censored in our observation scheme (we observe their lifetime only up to time $T$) but contribute to the whole estimation process at the same level as the population $\mathring{\mathcal T}_T \subset \mathcal T_T$ of individuals born and dead before $T$: due to the supercriticality of the process ($m >1$) we have $|\mathcal T_T| \approx |\mathring {\mathcal T}_T| \approx |\partial \mathcal T_T|$ as $T$ grows to infinity, and this affects the statistical analysis, see  Section~\ref{first estimates} below.
\item[{\bf 3)}] {\it Non-ancillarity:} the number of observations $|\mathcal T_T|$ that governs the amount of statistical information is random and its distribution depends on $B$: we essentially have less observations if $B$ is small (particles split at a slow rate) than if $B$ is large (particles split at a fast rate). This means that $|\mathcal T_T|$ is not ancillary in the terminology of Fisher: it is not possible to ignore its randomness (by conditioning upon its value for instance) without losing some statistical information. We refer to the Encyclopedia of Statistics \cite{Fisherancillary} for more details.   
\end{itemize}
\subsubsection*{Main results}
We first study in Section~\ref{resultat proba} the behaviour of empirical measures of the form
$$\mathcal E^T(\mathcal V,g) = 
 |\mathcal V|^{-1} \sum_{u \in \mathcal V}g(\zeta_u^{T}),\;\;\text{with}\;\;\mathcal V = \mathring{\mathcal T_T}\;\;\text{or}\;\;\partial \mathcal T_T$$
for suitable test functions $g$. From the classical study of critical branching processes, it is known that $|\mathring{\mathcal T_T}| \approx |\partial \mathcal T_T| \approx e^{\lambda_BT}$, where $\lambda_B>0$ is the Malthus parameter associated to the model (Harris~\cite{Harris} and \eqref{def malthus} below). Both $\mathcal E^T(\mathring{\mathcal T_T},g)$ and $\mathcal E^T(\partial \mathcal T_T,g)$ converge to their respective limits with rate  $\exp(-\lambda_BT/2)$, with some uniformity  in $B$ and $g$ as shown in Theorem~\ref{rate en T} and~\ref{rate avant T} below. For the proof, we heavily rely on the recent studies of Cloez~\cite{cloez} and Bansaye {\it et al.}~\cite{BDMT}, two key references for this paper, adjusting the tools developed in~\cite{BDMT} to the non-Markovian case: the essential ingredient is the use of many-to-one formulae that reduce the problem to studying the evolution of a particle picked at random along the genealogical tree (Propositions~\ref{prop: Mt1} and~\ref{prop: Mt1 forks}). The rate of convergence to equilibrium of this tagged particle, which governs the rates of convergence for statistical estimators, is obtained by a simple coupling argument (Proposition~\ref{prop: convergence semigroupe}). 

\vip

These preliminary results enable us to address the main issue of the paper: we construct in Section~\ref{sec: stat estimation} a nonparametric estimator $\widehat B_T(x)$ of $B(x)$ that achieves the rate of convergence $\exp(-\lambda_B\frac{\beta}{2\beta+1}T)$ for pointwise error and uniformly over  functions $B$ with local smoothness of order $\beta>0$ (Theorem~\ref{thm: upper rate}). We show that this rate is optimal in a minimax sense in Theorem~\ref{thm: borne inf}, thanks to statistical tools developed in L\"ocherbach~\cite{eva1}.  This result is obtained under the restriction that convergence to equilibrium of a tagged particle is faster than the growth of the tree. Otherwise, we still have a rate of convergence, but we do not have (nor believe in) its optimality. We bypass the aforementioned bias selection difficulty 1) by weighting a kernel estimator  by a de-biasing factor that depends on preliminary estimators of $\lambda_B$ and $m$. These estimators (essentially) converge with rate $\exp(-\lambda_BT/2)$  as shown in Proposition~\ref{prop: conv lambda}. As for the censoring part 2), we base our nonparametric kernel estimator on $\mathcal E^T(\mathring{\mathcal T_T},g)$ and {\it not on}  $\mathcal E^T(\partial \mathcal T_T,g)$, since that latter quantity would lead to a subobtimal rate of convergence as discussed in Section~\ref{discussion}. Finally, the non-ancillarity issue 3) is solved by specifying a random bandwidth for the kernel that also depends on the preliminary estimation of $\lambda_B$. This last point requires extra efforts in order to show a form of stability that is detailed in Proposition~\ref{lem: kolmo}.

\vip

The statistical study of branching processes goes back to Athreya and Keiding~\cite{AthreyaKeiding} for deriving maximum likelihood theory in the case of a parametric (constant) division rate, relying on the fact that the number of living cells is then a Markov process, a property we lose here for a non-constant division rate $x \leadsto B(x)$. The textbook of Guttorp~\cite{guttorp} gives an account of existing parametric methods in the 1990's. In the early 2000's the regularity in the sense of the LAN and LAMN property was established in the comprehensive study of L\"ocherbach~\cite{eva1, eva2}, see also Hyrien~\cite{Hyrien} for statistical computational methods and Johnson {\it et al.}\;\cite{JSvR} for Bayesian analysis, and Delmas and Marsalle~\cite{DM} in discrete time. In nonparametric estimation, only few results exist; we mention the case when dynamics between jumps is driven by a diffusion in H\"opfner {\it et al.}~\cite{HHL}. To the best or our knowledge, our study provides with the first fully nonparametric approach in continuous time in supercritical branching processes which are piecewise deterministic. Admittedly, the Bellman-Harris model is a toy model for the study of population dynamics, but we believe that the present contribution sheds some light in the intrinsic difficulties that need to be solved in more elaborate models like cell equation for which only simplified statistical models have been considered so far (in discrete time or under additional deterministic or stochastic noise like in {\it e.g.}~\cite{DPZ, DHRR, DHKR1}). 
Concerning bias selection, density estimation when observing a biased sample has been studied at length framework by Efromovich \cite{Efr}.

\subsubsection*{Organisation of the paper}

In Section~\ref{resultat proba}, we define our rigorous statistical framework by means of continuous time rooted trees (Section~\ref{cont time rooted trees}) and study the convergence properties of the biased empirical measures $\mathcal E^T(\mathring{\mathcal T_T},g)$ and  $\mathcal E^T(\partial \mathcal T_T,g)$ in Section~\ref{conv emp measures}.  We start by deriving heuristically the respective limits of the empirical measures in Section~\ref{first estimates} (that can also be found in Cloez~\cite{cloez} and Bansaye {\it et al.}\;\cite{BDMT}) in order to shed some light on the specific methods of proof in the subsequent study of rate of convergence. We construct in Section~\ref{sec: stat estimation}  the estimators of $m$, $\lambda_B$ and $B(x)$ and state our statistical results together with a discussion on the extensions and limitations of our findings. Section~\ref{section: simus} tackles the problem of numerical implementation on simulated data, advocating for a reasonably use of our estimators in practice. Section~\ref{sec: proofs} is devoted to the proofs. An appendix (Section~\ref{sec: appendix}) contains auxiliary useful results.

\section{\textsc{Rate of convergence for biased empirical measures}} \label{resultat proba}

\subsection{Continuous time rooted trees} \label{cont time rooted trees}

It will prove more convenient to work with a representation of $(X(t))_{t \geq 0}$ in terms of a continuous time rooted tree. We need some notation and closely follow Bansaye {\it et al.}~\cite{BDMT}.
Let
$$\mathcal U=\bigcup_{k \geq 0}(\mathbb N^\star)^k$$
with $\mathbb N^\star = \{1, 2, \ldots\}$ and $(\mathbb N^\star)^0=\{\varnothing\}$ denote the infinite genealogical tree. We use throughout the following standard notation: for $u=(u_1, u_2, \ldots, u_m)$ and $v=(v_1,\ldots, v_n)$ in $\mathcal U$, we write $uv=(u_1,\ldots, u_m,v_1,\ldots, v_n)$ for the concatenation, we identify $\varnothing u$, $u \varnothing $ and $u$, we write $u \preceq v$ if there exists $w$ such that $uw=v$ and $u \prec v$ if $u \preceq v$ and $w \neq \varnothing$. For $u=(u_1, u_2, \ldots, u_m)$, we also write $|u| = m$.

Given a family $(\nu_u, u \in \mathcal U)$ of integers 
representing the number of children of the individuals $u \in \mathcal U$, we construct an ordered rooted tree  $\mathcal T \subset \mathcal U$ as follows:
\begin{itemize}
\item[i)] $\varnothing \in \mathcal T$,
\item[ii)] If $v \in \mathcal T$, $u \preceq v$ implies $u \in \mathcal T$,
\item[iii)] For every $u \in \mathcal T$, we have $uj \in \mathcal T$ if and only if $1 \leq j \leq \nu_u$.
\end{itemize}
For a family $(\zeta_u, u \in \mathcal U)$ of nonnegative numbers representing the lifetimes of the individuals $u \in \mathcal U$, we set
\begin{equation} \label{eq:birth and death}
b_u = \sum_{v \prec u}\zeta_v\;\;\text{and}\;\;d_u=b_u+\zeta_u
\end{equation}
for the times of birth and death of the individual $u \in \mathcal U$. 
Let  $\mathbb U = \mathcal U \times [0,\infty).$
A continuous time rooted tree is then a subset $\mathbb T$ of $\mathbb U$ such that
\begin{itemize}
\item[(i)] $(\varnothing, 0) \in \mathbb T$,
\item[(ii)] The projection $\mathcal T$ of $\mathbb T$ on $\mathcal U$ is an ordered rooted tree,
\item[(iii)] There exists a family $(\zeta_u, u \in \mathcal U)$ of nonnegative numbers such that $(u,s) \in \mathbb T$ if and only if $b_u \leq s < d_u$, where $(b_u,d_u)$ are defined by \eqref{eq:birth and death}.
\end{itemize}
We now work on some probability space $(\Omega, \mathcal F, \mathbb P)$. In this setting, we have the following
\begin{defi}[The Bellman-Harris model] \label{BellmanHarris}
A random continuous time rooted tree is a Bellman-Harris model with offspring distribution $p=(p_k)_{k \geq 1}$ and division rate $B:[0,\infty)\rightarrow [0,\infty)$ if
\begin{enumerate}
\item[(i)] The family of the number of children $(\nu_u, u \in \mathcal U)$ are independent random variables with common distribution $p$. 
\item[(ii)] The family of lifetimes $(\zeta_u, u \in \mathcal U)$ are independent random variables such that
\begin{equation} \label{def B}
\mathbb P\big(\zeta_u\geq x\big)=\exp\big(-\int_0^xB(y)dy\big),\;\;x \geq 0,
\end{equation}
with $$\int^\infty B(x)dx=\infty,$$
\item[(iii)] The families of random variables $(\nu_u, u \in \mathcal U)$ and $(\zeta_u, u \in \mathcal U)$ are independent.
\end{enumerate}
\end{defi}
Going back to the process $(X(t))_{t \geq 0}$ defined in \eqref{eq:def X},we have an identity between point measures on $(0,\infty)$ that reads
$$\sum_{i \geq 1}{\bf 1}_{\{X_i(t)>0\}}\delta_{X_i(t)} = \sum_{u \in \mathcal T}{\bf 1}_{\{t \in [b_u,d_u)\}}\delta_{t-b_u}.$$

The following assumption will be in force in the paper:
\begin{assumption} \label{hyp: offspring}
The offspring distribution $p = (p_k)_{k \geq 0}$ satisfies
$$p_0=p_1=0,\;\;2 \leq m=\sum_{k \geq 2}kp_k<\infty, \quad \sum_{k \geq 2}k^2p_k<\infty\;\;\text{and}\;\; \bar{m} = \sum_{i\neq j}\sum_{k \geq i \vee j}p_k <\infty.$$
\end{assumption}

The technical condition $\bar{m}  < \infty$ is needed for the so-called many-to-one formulae, see Proposition~\ref{prop: Mt1 forks} below.

\subsection{The limiting objects} \label{first estimates}
In order to extract information about $x \leadsto B(x)$, we consider the empirical distribution function over the lifetimes indexed by some $\mathcal V_T \subset \mathcal T_T$ for a test function $g$, that is 
$$\mathcal E^T(\mathcal V_T,g) = |\mathcal V_T|^{-1}\sum_{u \in \mathcal V_T}g(\zeta_u^{T}),$$
and expect a law of large number as $T \rightarrow \infty$. Without much of a surprise, it turns out that depending whether $\zeta_u^{T} = \zeta_u$ or not, {\it i.e.} if the data are still alive at time $T$, therefore censored or not, we have a different limit. More precisely, define 
$$\mathring {\mathcal T}_T = \{u \in \mathcal T,b_u<T\;\text{and}\;d_u\leq T\}\;\;\text{and}\;\;\partial\mathcal T_T = \{u \in \mathcal T,b_u \leq T < d_u\},$$
{\it i.e.} the set of particles that are born and that die before $T$, and the set of particles alive at time $T$, so that 
$\mathcal T_T = \mathring {\mathcal T}_T \cup \partial\mathcal T_T.$
We need some notation. Introduce the {\it Malthus parameter} $\lambda_B >0$ defined as the (necessarily unique) solution to 
\begin{equation} \label{def malthus}
\int_0^\infty B(x)e^{-\lambda_{B}x-\int_0^xB(y)dy}dx=\frac{1}{m}.
\end{equation}
To a division rate function $x \leadsto B(x)$ satisfying the properties of Definition \ref{BellmanHarris}, we associate its {\it density lifetime}
$$f_B(x) = B(x)\exp\big(-\int_0^x B(y)dy\big),\;x \geq 0$$
and its biased density lifetime
$$f_{H_B}(x) = me^{-\lambda_B x}f_B(x),\;x\geq 0,$$
which in turns uniquely defines a biased division rate
\begin{equation} \label{biased division rate}
H_B(x)=\frac{me^{-\lambda_Bx}f_B(x)}{1-m\int_0^xe^{-\lambda_By}f_B(y)ds}.
\end{equation}
Finally, we define the limiting measures
\begin{equation} \label{def limite bord}
 \partial \mathcal E_B(g) = \lambda_B\frac{m}{m-1} \int_0^\infty g(x)e^{-\lambda_Bx} e^{-\int_0^x B(y)dy} dx 
\end{equation}
and
\begin{equation} \label{def limite int}
\mathring{\mathcal E}_B\big(g\big)  =  m\int_0^\infty g(x)e^{-\lambda_B x}f_{B}(x)dx = \int_0^\infty g(x)f_{H_B}(x)dx.
\end{equation}
It is known that $\mathcal E^T(\partial \mathcal T_T,g) \rightarrow \partial \mathcal E_B(g)$ and $\mathcal E^T(\mathring {\mathcal T}_T,g) \rightarrow \mathring{\mathcal E}_B\big(g\big)$ in probability as $T \rightarrow \infty$, see Appendix \ref{heuristic} for heuristics and references.  We establish in Theorems \ref{rate en T} and \ref{rate avant T} in the next Section \ref{conv emp measures} a rate of convergence
with some uniformity in $B$. The rate is linked to $\lambda_B$ and the geometric ergodicity of an auxiliary one-dimensional Markov process with infinitesimal generator 
\begin{equation} \label{def biased generator}
\mathcal A_{H_B} g(x) = g'(x)+H_B(x)\big(g(0)-g(x)\big)
\end{equation}
densely defined on continuous functions vanishing at infinity and that represents the value of a branch along the tree picked uniformly at random at each branching event. 
\subsection{Convergence results for biased empirical measures} \label{conv emp measures}
\subsubsection*{Notation.} 
For constants $b,C>0$, introduce the sets
$$\mathcal L_C=\Big\{g:[0,\infty)\rightarrow \R,\; \sup_x|g(x)|\leq C\Big\}$$
and
$$\mathcal B_{b,C} = \Big\{B:[0,\infty) \rightarrow [0,\infty),\forall x\geq 0:\;b\leq B(x) \leq b\max\{C,1\} \Big\}.$$
%
For a family $\Gamma_T = \big(\Gamma_T(\gamma)\big)_{T \geq 0}$ of real-valued random variables, with distribution depending on some parameter $\gamma \in \mathcal G$ we say that $\Gamma_T$ is $\mathcal G$-tight for the parameter $\gamma$ if
$$\sup_{T >0, \gamma \in \mathcal G}\PP\big(|\Gamma_T(\gamma)| \geq K\big) \rightarrow 0\;\;\text{as}\;\;K\rightarrow \infty.$$
\subsubsection*{Results} We have a trade-off between the growth rate $\lambda_B$ of the tree $\E[|\mathcal T_T|] \approx e^{\lambda_B T}$ and the convergence to equilibrium of the Markov process with infinitesimal generator $\mathcal A_{H_B}$ defined in \eqref{def biased generator} above. More, precisely, we show in Proposition~\ref{prop: convergence semigroupe} below the estimate
$$
\Big|P_{H_B}^tg(x)-\int_0^\infty g(y)\mu_B(y)dy\Big| \leq 2 \sup_y|g(y)| e^{-\rho_B t}\;\;\text{for every}\;x \in (0,\infty).
$$
Here, $(P_{H_B}^t)_{t \geq 0}$ denotes the semigroup associated to $\mathcal A_{H_B}$ and $\mu_B$ its unique invariant probability, and
$$\rho_B = \inf_x H_B(x)$$
where $H_B(x)$ is the biased division rate defined in \eqref{biased division rate} above. 
The rate of convergence of the biased empirical measures  $\mathcal E^T(\mathring{\mathcal T_T},g)$ and  $\mathcal E^T(\partial \mathcal T_T,g)$ to their limits $\partial {\mathcal E}_B(g)$ and $\mathring{\mathcal E}_B(g)$ respectively defined by \eqref{def limite bord} and \eqref{def limite int} are goverened by $\lambda_B$ and $\rho_B$: define
\begin{equation} \label{def rate base}
v_T(B) =
\left\{
\begin{array}{lll}
e^{-\min\{\rho_B,\lambda_B/2\} T} & \text{if} & \lambda_B \neq 2\rho_B, \\ \\
T^{1/2} e^{-\lambda_BT/2} &  \text{if} & \lambda_B = 2\rho_B.
\end{array}
\right.
\end{equation}
We have:
\begin{thm}[Rate of convergence for particles living at time $T$] \label{rate en T}
Work under Assumption~\ref{hyp: offspring}. For every $b, C, C' >0$, 
$$v_T(B)^{-1}\big(\mathcal E^T\big(\partial{\mathcal T}_T,g\big) - \partial {\mathcal E}_B(g)\big)$$
is $\mathbb B_{b,C} \times \mathcal L_{C'}$-tight for the parameter $(B,g)$.
\end{thm}

\begin{thm}[Rate of convergence for particles dying before $T$] \label{rate avant T} In the same setting as Theorem~\ref{rate en T},
$$v_T(B)^{-1}\big(\mathcal E^T\big(\mathring {\mathcal T}_T,g\big) - \mathring{\mathcal E}_B(g)\big)$$
is $\mathbb B_{b,C} \times \mathcal L_{C'}$-tight for the parameter $(B,g)$.
\end{thm}
Several comments are in order: 


\subsubsection*{About the rate of convergence and the class $\mathcal B_{b,C}$:} the restriction  $B \in \mathcal B_{b,C}$ enables us to obtain uniform convergence results. This is important for the subsequent statistical analysis. However, this can be relaxed if only $\mathcal L_{C'}$-tightness is sought, provided $B$ complies to the conditions of Definition \ref{BellmanHarris} and Assumption~\ref{hyp: offspring} and $\rho_B>0$. In the same direction, the rate $v_T(B)$ can be improved replacing $\rho_B = \inf_x H_B(x)$ in \eqref{def rate base} by
\begin{equation} \label{optimal geo}\rho_B^\star = \sup\Big\{\rho,\;\forall x, t>0: |P_{H_B}^tg(x)-\int_0^\infty g(y)\mu_B(y)dy| \leq 2 \sup_y|g(y)| e^{-\rho t}\Big\}, 
\end{equation}
and we have in particular $\rho_B^\star \geq \rho_B$.
\subsubsection*{About the tightness:} what we need in order to handle the random normalisation in $\mathcal E^T\big(\mathring {\mathcal T}_T,g\big)$  is actually the convergence of $e^{\lambda_B T}|\mathring {\mathcal T}_T|^{-1}$.  This convergence still holds in probability but not necessarily in $L^2(\PP)$, so we only have tightness in Theorems~\ref{rate en T} (and~\ref{rate avant T} for the same reason). However, if we replace $\mathcal E^T(\mathring {\mathcal T}_T, g)$ by
$$\frac{1}{\E[|\mathring {\mathcal T}_T|]} \sum_{u \in \mathring {\mathcal T}_T}g(\zeta_u^{T}) ,$$
then we have a bound in $L^2(\PP)$ together with a control on $g$, see Proposition~\ref{lem: vanishing test} below. Such a finer control is mandatory for the subsequent statistical analysis, since we need to pick a function $g$ that depends on $T$ and that mimics the behaviour of the Dirac mass $\delta_x$, see Section~\ref{sec: stat estimation} below.

\section{\textsc{Statistical estimation}} \label{sec: stat estimation}

\subsection{Construction of an estimation procedure}
\subsubsection*{Estimation of $m$ and $\lambda_B$} 
To a particle sitting at node $u \in \mathring{\mathcal T}_T$, we associate its number of children $\nu_u$ (see Definition~\ref{BellmanHarris}). Note that the knowledge of ${\mathcal T}_T$ enables us to reconstruct $\nu_u$ for every $u \in \mathring{\mathcal T}_T$. This enables us to define an estimator for $m$ by setting
\begin{equation} \label{est m}
\widehat m_T =|\mathring{\mathcal T}_T|^{-1}\sum_{u \in \mathring{\mathcal T}_T}\nu_u
\end{equation}
on the set $ |\mathring{\mathcal T}_T| \neq 0 $ and $2$ otherwise.
In order to estimate $\lambda_B$, we first observe that for $\mathrm{Id}(x)=x$, we can write 
\begin{align*}
\mathring{\mathcal E}_B(\mathrm{Id}) & = m\int_0^\infty x\big(B(x)+\lambda_B\big)e^{-\int_0^x (B(y)+\lambda_B)dy}dx - m\lambda_B \int_0^\infty xe^{-\lambda_Bx}e^{-\int_0^xB(y)dy}dx \\
& = m \int_0^\infty e^{-\int_0^x (B(y)+\lambda_B)dy} dx - m \lambda_B \tfrac{m-1}{m\lambda_B} \partial {\mathcal E}_B(\mathrm{Id}) 
 = m\tfrac{m-1}{m\lambda_B} - (m-1)\partial {\mathcal E}_B(\mathrm{Id}), 
\end{align*}
the last equality being obtained integrating by parts. So we obtain the following representation
$$\lambda_B = \Big(\tfrac{1}{m-1}\mathring{\mathcal E}_B(\mathrm{Id})+\partial {\mathcal E}_B(\mathrm{Id}) \Big)^{-1}$$
and this yields the estimator
\begin{equation} \label{est lambda}
\widehat \lambda_T = \Big(\tfrac{1}{\widehat m_T-1}|\mathring{\mathcal T}_T|^{-1}\sum_{u \in \mathring {\mathcal T}_T}\zeta_u+|\partial {\mathcal T}_T|^{-1}\sum_{u \in \partial{\mathcal T}_T}\zeta_u^T\Big)^{-1}.
\end{equation}
The following convergence result for $\widehat \lambda_T$ is then a consequence of Theorems~\ref{rate en T} and~\ref{rate avant T}.
\begin{prop} \label{prop: conv lambda} In the same setting as Theorem~\ref{rate en T} with $v_T(B)$ given in \eqref{def rate base} above, we have that
$$e^{\lambda_B T/2}\big(\widehat m_T - m\big)\;\;\text{and}\;\; T^{-1}v_T(B)^{-1}\big(\widehat \lambda_T - \lambda_B\big)$$
are $\mathcal B_{b,C}$-tight for the parameter $B$.
\end{prop}

\subsubsection*{Reconstruction formula for $B(x)$}
An {\it estimator} $\widehat B_T:[0,\infty)\rightarrow \R$ of $B$ is a random function 
$$\widehat B_T(x)= \widehat B_T\big(x, (X(t))_{t \in [0,T]}), \;\;x\in [0,\infty)$$
that is measurable as a function of $(X(t))_{t \in [0,T]}$ but also as a function of $x$. By \eqref{hazard formula}, we have 
$$
B(x) = \frac{f_B(x)}{1-\int_0^x f_B(y)dy} $$ and from the definition
$\mathring{\mathcal E}_B\big(g\big) = m\int_0^\infty g(x) e^{-\lambda_Bx}f_B(x)dx$ we obtain the formal reconstruction formula
\begin{equation} \label{formal reconstruction}
B(x) =  \frac{\mathring{\mathcal E}_B\big(m^{-1}e^{\lambda_B\cdot}\delta_x(\cdot)\big)}{1- \mathring{\mathcal E}_B\big(m^{-1}e^{\lambda_B\cdot}{\bf 1}_{\{\cdot \leq x\}}\big)}
\end{equation}
where $\delta_x(\cdot)$ denotes the Dirac function at $x$. Therefore, substituting $m$ and $\lambda_B$ by the estimators defined in \eqref{est m} and \eqref{est lambda} and taking $g$ as a weak approximation of  $\delta_x$, we obtain a strategy for estimating $B(x)$ replacing furthermore $\mathring{\mathcal E}_B(\cdot)$ by its empirical version ${\mathcal E}^T(\mathring{\mathcal T}_T, \cdot)$.

\subsubsection*{Construction of a kernel estimator and function spaces}
Let $K: \R \rightarrow \R$ be a kernel function. For $h>0$, set $K_h(x)=h^{-1}K(h^{-1}x)$. In view of \eqref{formal reconstruction}, we define the estimator
\begin{align*}
\widehat B_T(x) & = \frac{{\mathcal E}^T\big(\mathring{\mathcal T}_T, \widehat m_T^{-1}e^{\widehat \lambda_T \cdot }K_h(x-\cdot)\big)} 
{1-{\mathcal E}^T\big(\mathring{\mathcal T}_T, \widehat m_T^{-1}e^{\widehat \lambda_T \cdot} {\bf 1}_{\{\cdot \leq x\}}\big)}
\end{align*}
on the set ${\mathcal E}^T\big(\mathring{\mathcal T}_T, \widehat m_T^{-1}e^{\widehat \lambda_T \cdot} {\bf 1}_{\{\cdot \leq x\}}\big) \neq 1$ and $0$ otherwise. Thus $\widehat B_T(x)$ is specified by the choice of the kernel  $K$ and the bandwidth $h>0$. Note that the observations $(\zeta_u, u \in \partial \mathcal T_T)$ only occur in the estimator $\widehat \lambda_T$ of $\lambda_B$. \\ 

We need the following property on $K$:

\begin{assumption} \label{kernel} The kernel $K:\R \rightarrow \R$ is differentiable with compact support and for some integer $n_0 \geq 1$, we have $\int_{-\infty}^\infty x^kK(x)dx= {\bf 1}_{\{k=0\}}$ for $k=1,\ldots, n_0$.
\end{assumption} 

Assumption~\ref{kernel} will enable us to have nice approximation results over smooth functions $B$, described in the following way: for a compact interval $\mathcal D \subset (0,\infty)$ and $\beta>0$, with $\beta=\lfloor \beta\rfloor + \{\beta\}$, $0< \{\beta\} \leq 1$ and $\lfloor \beta\rfloor$ an integer, let $\mathcal H_\mathcal D^\beta$ denote the H\" older space of functions $g:{\mathcal D}\rightarrow \R$ possessing a derivative of order $\lfloor \beta \rfloor$ that satisfies
\begin{equation} \label{def sob}
|g^{\lfloor \beta \rfloor}(y)-g^{\lfloor \beta \rfloor}(x)| \leq c(g)|x-y|^{\{\beta\}}.
\end{equation}
The minimal constant $c(g)$ such that \eqref{def sob} holds defines a semi-norm $|g|_{{\mathcal H}_\mathcal D^\beta}$. We equip the space ${\mathcal H}^\beta_\mathcal D$ with the norm 
$\|g\|_{{\mathcal H}^\beta_\mathcal D} = \sup_{x}|g(x)|+ |g|_{{\mathcal H}_\mathcal D^\beta}$ and the balls
$${\mathcal H}_\mathcal D^\beta(L) = \{g:\mathcal D\rightarrow \R,\;\|g\|_{{\mathcal H}_\mathcal D^\beta} \leq L\},\;L>0.$$

\subsection{Convergence results for $\widehat B_T(x)$} \label{convergence statistique}
We are ready to give our main result, namely
a rate of convergence of $\widehat B_T(x)$ for $x$ restricted to a compact interval $\mathcal D$, uniformly over H\" older balls ${\mathcal H}_\mathcal D^\beta(L)$ of (known) smoothness $\beta$ intersected with $\mathcal B_{b,C}$. Define
\begin{equation} \label{def rate esti}
w_T(B) =
T^{{\bf 1}_{\{\lambda_B=2\rho_B\}}} \exp\big(-\min\{\lambda_B, 2\rho_B\}\frac{\beta-(\lambda_B/\rho_B-1)_+/2}{2\beta+1}T\big)
\end{equation}
and note that when $\rho_B \geq \lambda_B$, we have $w_T(B) = e^{-\lambda_B \frac{\beta}{2\beta +1}T} \approx \E[|\mathcal T_T|]^{-\beta/(2\beta+1)}$.

\begin{thm}[Upper rate of convergence] \label{thm: upper rate}
Specify $\widehat B_T$ with a kernel satisfying Assumption~\ref{kernel} for some $n_0>1$ and
\begin{equation} \label{def bandwidth}
h=\widehat h_T = \exp\big(- \widehat \lambda_T \tfrac{1}{2\beta+1} T \big)
\end{equation}
for some $\beta \in [1/2,n_0)$. For every $b,C >0, L>0$, every compact interval $\mathcal D$ in $(0,\infty)$ (with non-empty interior) and every $x \in \mathcal D$,
$$w_T(B)^{-1}\big(\widehat B_T(x)-B(x)\big)$$
is $\mathcal B_{b,C} \cap \mathcal H^\beta_\mathcal D(L)$-tight for the parameter $B$.
\end{thm}
We have a partial optimality result in a minimax sense. Define
$$\mathcal B^{+}_{b,C} = \big\{B \in \mathcal B_{b,C},\; \lambda_B \leq \rho_B\big\}\;\;\text{and}\;\;\mathcal B^{-}_{b,C} =  \big\{B \in \mathcal B_{b,C},\; \rho_B \leq \lambda_B\big\}$$
so that $\mathcal B_{b,c} = \mathcal B^{+}_{b,C} \cup \mathcal B^{-}_{b,C}$
We then have the following

\begin{thm}[Lower rate of convergence over $\mathcal B_{b,C}^+$] \label{thm: borne inf}
Let $\mathcal D$ be a compact interval in $(0,\infty)$. For every $x \in \mathcal D$ and every positive $b, C, \beta, L$, there exists $C'>0$ such that 
$$\liminf_{T \rightarrow \infty}\inf_{\widehat B_T}\sup_{B }\PP\big(e^{\lambda_{B}\frac{\beta}{2\beta+1}T}\big|\widehat B_T(x)-B(x)\big| \geq C'\big) > 0,$$
where  the supremum is taken among all $B\in \mathcal B^{+}_{b,C} \cap \mathcal H^\beta_\mathcal D(L)$ and the infimum is taken among all estimators.
\end{thm}

We observe a conflict between the rate growth of the tree $\lambda_B$ and its convergence rate to equilibrium $\rho_B$. On $\mathcal B^{+}_{b,C}$ we retrieve the expected usual optimal rate of convergence $\exp(-{\lambda_B\tfrac{\beta}{2\beta+1}T}) \approx \E[|\mathcal T_T|]^{-\beta/(2\beta+1)}$ whereas if $\rho_B \leq \lambda_B$, we obtain the deteriorated rate $\exp\big(-\min\{\lambda_B, 2\rho_B\}(\beta-\tfrac{1}{2}(\tfrac{\lambda_B}{\rho_B}-1))/(2\beta+1)T\big)$ and this rate is presumably not optimal, as discussed at length in Section \ref{discussion} below.

\subsection{Discussion of the results} \label{discussion}


\subsubsection*{Rates of convergence} 
The ``parametric case" for a constant division rate $B(x)=b$ with $b>0$ has a statistical simpler structure, but also a nice probabilistic feature since the process $t \leadsto |\partial \mathcal T_t|$, {\it i.e.} the number of cells alive at time $t$ is Markov. In that setting, explicit (asymptotic) information bounds are available (Athreya and Keiding~\cite{AthreyaKeiding}). In particular, the model is regular with asymptotic Fisher information of order $e^{\lambda_BT}$, thus the best-achievable (normalised) rate of convergence is $e^{-\lambda_BT/2}$. This is consistent with the  
minimax rate $\exp(-{\lambda_B\tfrac{\beta}{2\beta+1}T})$ that we obtain for the class $\mathcal H^\beta_\mathcal D(L) \cap \mathcal B_{b,C}^+$, and we retrieve the parametric rate by formally setting $\beta=\infty$ in the previous formula.\\ 

However, this rate is strongly parameter dependent in the sense that it also depends on $B$ via $\lambda_B$. This dependence is severe, since it appears at the same level as the smoothness exponent $\beta/(2\beta+1)$ in the rate exponent $\frac{\beta}{2\beta+1}\lambda_B$. For instance, in the simplest case of a constant function $B(x)=b$ for every $x\geq 0$, we have $\lambda_B = (m-1)b$, and we see that $B$ ($b$ here) plays at the same level as $\beta/(2\beta+1)$. This also has a non-trivial technical cost in establishing rates of convergence for the estimator $\widehat B_T(x)$: in order to minimise the bias-variance tradeoff, the (log)-bandwidth has to be chosen as $-\lambda_B\frac{1}{2\beta+1}T\big(1+o(1)\big)$ exactly, and this is achieved by the plug-in rule $-\widehat \lambda_T\frac{1}{2\beta+1}T$ thanks to Proposition~\ref{lem: kolmo}. We then have to carefully check that our estimator is not too sensitive to this further approximation, and this requires the analysis of the smoothness of the process $h \leadsto \widehat B_{T,h}(x)$ where $h$ is the bandwidth of $\widehat B_T(x)$, as shown in Proposition~\ref{lem: kolmo}.

\subsubsection*{Fast convergence to equilibrium in $\mathcal B_{b,C}^+$ versus slow convergence in $\mathcal B_{b,C}^-$} While we have an optimal rate of convergence over $\mathcal B_{b,B}^+$, the situation is unclear over $\mathcal B_{b,C}^-$. First, the convergence rate to equilibrium $\rho_B$ should be replaced by an estimator and that would lead to extraneous difficulties. Even if we knew $\rho_B$, optimising the bias-variance trade-off in the proof of Theorem~\ref{thm: upper rate} would not lead to the expected rate $\exp(-\min\{\lambda_B, 2\rho_B\}\tfrac{\beta}{2\beta+1}T)$ but to an intermediate rate that reads
\begin{equation} \label{intermediate vitesse}
\exp\big(-\min\{\lambda_B, 2\rho_B\}\frac{\min\{\max\{\rho_B/\lambda_B, 1/2\}, 1\}\beta}{2\min\{\max\{\rho_B/\lambda_B, 1/2\}, 1\}\beta+1}T\big), 
\end{equation}
and that continuously deteriorates as $\rho_B$ separates $\lambda_B$ from below. Let us also mention that the classes $\mathcal B_{b,C}^+$ and $\mathcal B_{b,C}^-$ are never trivial. To that end, define
\begin{equation} \label{def top smooth}
\mathbb B_{b,m} = \big\{B \in \mathcal B_{b,m/(m-1)},\;\forall x\geq 0:B'(x)-B(x)^2 \leq 0\big\}
\end{equation}
where $m=\sum_{k \geq 2}kp_k$ is the mean number of children at each branching event. 
\begin{prop} \label{classe B plus}  For any $b > 0$, we have $\mathbb B_{b,m} \subset \mathcal B^+_{b,m/(m-1)}$. For every $C >2m(m+2)b/(m-1), \beta >0$ and any compact interval $\mathcal D \subset (0,\infty)$, there exists $B \in \mathcal H_{\mathcal D}^\beta$ such that $B \in \mathcal B_{b,C}^-$ and $B  \notin \mathcal B_{b,C}^+$.
\end{prop}
In the proof of Proposition \ref{classe B plus} below we show a versatility in the choice of functions $B$ that yield either fast or slow rate of convergence to equilibrium. Finally, one could (at least formally) replace $\rho_B$ by $\rho_B^\star$, the optimal geometric rate of convergence to equilibrium defined in\eqref{optimal geo} above,  but that would only improve on the rate of convergence \eqref{intermediate vitesse} replacing $\rho_B$ by $\rho_B^\star$ which we do not know how to estimate, neither analytically nor statistically and the obtained result would still presumably not be optimal. This suggests a totally different estimation strategy -- that we do not have at the moment -- whenever convergence to equilibrium is slow. 
\subsubsection*{Other loss functions} If $\mathcal K \subset \mathring{\mathcal D}$ is a closed interval ($\mathring {\mathcal D}$ denotes the interior of $\mathcal D$), then Theorem~\ref{thm: upper rate} also holds uniformly in $x \in \mathcal K$. So we also have that 
$$w_T(B)^{-2}\int_{\mathcal K}(\widehat B_T(x)-B(x)\big)^2dx$$
is $\mathcal B_{b,C} \cap \mathcal H^\beta_\mathcal D(L)$-tight for the parameter $B$. For integrated squared error-loss, we could weaken the smoothness constraint $B \in {\mathcal H}_\mathcal D^\beta(L)$ to Sobolev smoothness (see {\it e.g.}~\cite{Tsybakov}) when the smoothness is measured in $L^2$-norm. An extension of Theorem~\ref{thm: borne inf} can be obtained likewise.

\subsubsection*{Smoothness adaptation} Our estimator $\widehat B_T(x)$ {\it is not} $\beta$-adaptive, in the sense that the choice of the $\mathcal B_{b,C}^+$-optimal (log) bandwidth $-\widehat \lambda_T\frac{1}{2\beta+1}T$ still depends on $\beta$, which is unknown in principle. In the numerical implementation Section~\ref{section: simus} below, we address this issue from a practical point of view. However, a theoretical result is still needed. The classical analysis of adaptive (or other) kernel methods \`a la Lepski for instance shows that this boils down to proving concentration inequalities of the type
\begin{equation} \label{concentration}
\PP\big(\big|\mathcal E^T\big(\mathring{\mathcal T}_T,g_h\big)- \mathring{\mathcal E}_B(g_h)\big| \geq e^{\lambda_BT/2} c(q,T)\big) \leq e^{-q\lambda_BT},\;\;q>0,
\end{equation}
where, for $0 < h^{-1} \leq e^{\lambda_BT}$, the test function  $g_h$ has the form $g_h(y) = h^{-1/2}g\big(h^{-1}(x-y)\big)$ with $x \in \mathcal D$ and $g \in \mathcal L_C$. The threshold $c(q,T)$ should be of order $q \lambda_B T$ and would inflate the risk by a slow term  (of order T). By a suitable choice of $q$, it would then be possible to obtain adaptation for $\beta$ in compact intervals. Concentration inequalities like \eqref{concentration} have been explored in~\cite{valere} in discrete time. To the best of our knowledge, such inequalities are not yet available in continuous time and lie beyond the scope of the paper.

\subsubsection*{Information from $\mathring{\mathcal T}_T$ versus $\partial{\mathcal T}_T$}
In the regime $B \in \mathcal B_{b,C}^+$, having
$$\partial \mathcal E_B(g) = \lambda_B\frac{m}{m-1} \int_0^\infty g(x)e^{-\lambda_Bx}\exp\big(-\int_0^x B(y)dy\big)dx $$
and ignoring the fact that the constants $m$ and $\lambda_B$ are unknown (or rather knowing that they can be estimated at the superoptimal rate $e^{\lambda_BT/2}$), we can anticipate that by picking a suitable test function $g$ mimicking a delta function $g(x) \approx \delta_x$, the information about $B(x)$ can only be inferred through $\exp(-\int_0^x B(y)dy)$, which imposes to further take a derivative hence some ill-posedness.\\

We can briefly make all these arguments more precise (still in the regime $B \in \mathcal B_{b,C}^+$) : we assume that we have estimators of $\widehat m_T$ of $m$ and $\widehat \lambda_T$ of $\lambda_B$  (using the ones defined in \eqref{est m} and \eqref{est lambda} or by any other means) that converge with rate $T^{-1}e^{\lambda_BT/2}$ as in Proposition~\ref{prop: conv lambda}. Consider the quantity
$$\widehat f_{h,T}(x) = - \mathcal E^T\Big(\partial \mathcal T_T, \frac{\widehat m_T-1}{\widehat \lambda_T\widehat m_T}\big(K_{h}\big)'(x-\cdot)\Big)$$
for a kernel satisfying Assumption~\ref{kernel}.
By Theorem~\ref{rate en T} and integrating by part, we readily see that 
\begin{equation} \label{conv ill posed}
\widehat f_{h,T} \rightarrow -\partial \mathcal E_B\Big(\frac{m-1}{\lambda_B m}\big(K_{h}\big)'(x-\cdot)\Big)=\int_0^\infty K_h(x-y)f_{B+\lambda_B}(y)dy
\end{equation}
in probability as $T \rightarrow \infty$, where $f_{B+\lambda_B}$ is the density associate to the division rate $x \leadsto B(x)+\lambda_B$. On the one hand, it is not difficult to show that Proposition~\ref{lem: vanishing test} (used in the proof of Theorem~\ref{thm: upper rate} below) is valid when substituting $\mathring{\mathcal T}_T$ by $\partial \mathcal T_T$, so we expect (altough not formally established) the rate of convergence in \eqref{conv ill posed} to be of order $h^{-3/2}e^{\lambda_BT/2}$ since we take the derivative of the kernel $K_h$. 
 On the other hand, the limit  $\int_0^\infty K_h(x-y)f_{B+\lambda_B}(y)dy$ approximates $f_{B+\lambda_B}(x)$ with an error of order $h^\beta$ if $B \in \mathcal H^\beta_\mathcal D.$ Balancing the two error terms in $h$, we see that we can estimate $f_{B+\lambda_B}(x)$ with an error of (presumably optimal) order $\exp(-\lambda_B\tfrac{\beta}{2\beta+3}T)$. Due to the fact that the denominator in representation \eqref{hazard formula} can be estimated with parametric error rate $\exp(-\lambda_BT/2)$  (possibly up to polynomially slow terms in $T$), we end up with the rate of estimation $\exp(-\lambda_B\frac{\beta}{2\beta+3}T)$ for $B(x)$ as well, and that can be related to an ill-posed problem of order 1 (see for instance~\cite{Tsybakov}).\\ 

This phenomenon, namely the structure of an ill-posed problem of order 1 in restriction to data alive at time $T$, has already been observed in other settings: for the estimation of a size-division rate from living cells at a given large time in Doumic {\it et al.}~\cite{DPZ, DHRR} or for the estimation of the dislocation measure for a homogeneous fragmentation in Hoffmann and Krell~\cite{HK}. 
Note also that this phenomenon does not appear in parametric estimation, since the number of data in $\mathring {\mathcal T}_T$ and $\partial \mathcal T_T$ are of the same order of magnitude (or put differently, the rates in Theorems~\ref{rate en T} and~\ref{rate avant T} are the same and govern the rate of estimation of a one dimensional parameter).

\section{\textsc{Numerical implementation}} \label{section: simus}

We assume that each cell $u \in \mathcal{U}$ has exactly two children at each division ($p_2 = 1$). This can model the evolution of a population of cells reproducing by binary divisions, as described deterministically by \eqref{det description}. We pick a trial division rate $B$ defined analytically by
$$
B(x) = 
\left\{
\begin{array}{lll}
\tfrac{1}{3}x^3-\tfrac{7}{8}x^2+\tfrac{5}{8}x+\tfrac{4}{10} & \text{if} & 0\leq x\leq\tfrac{3}{2} \\ \\
\tfrac{119}{160}-\tfrac{1}{4}\exp\big(-(x-\tfrac{3}{2})\big) & \text{if} & x > \tfrac{3}{2}
\end{array}
\right.
$$
and represented in Figure~\ref{fig: IC} (bold red line). We have $b\leq B(x) \leq \tfrac{m}{m-1}b$ for any $x\geq 0$ for $b = 0.4$ and $m=2$ and the lifetime density $f_B$ is non increasing (except in a vicinity of zero). Given $T> 0$ we simulate the lifetime of the rooted cell $\zeta_{\emptyset}$ with probability density $f_{B}$ and set $d_{\emptyset} = \zeta_{\emptyset}$. For $u  \in \mathcal{U}$ such that $d_{u} > T$, we do not simulate the lifetimes of its descendants since they are not in the observation scheme~$\mathring{\mathcal{T}}_T\cup\partial \mathcal{T}_T$.  For $u  \in \mathcal{U}$ such that $d_u \leq T$ we simulate $\zeta_{u0}$ and $\zeta_{u1}$ independently with probability density $f_{B}$; we set $d_{u0} : = d_u + \zeta_{u0}$ and $d_{u1} : = d_u + \zeta_{u1}$. Using R software, we generate $M = 100$ trees up to time $T =  23$, so that the mean number of observations $|\mathring{\mathcal{T}}_T|$ is sufficiently large. (Note that for a binary tree, we always have the identity $|\partial \mathcal{T}_T| = |\mathring{\mathcal{T}}_T|+1$.) Figure~\ref{Fig 1} represents a typical observation scheme with continuous or discrete representation. The (random) number of observations fluctuates a lot as shown in Table~\ref{tab : stat desc N_T} where some elementary statistics are given.

\begin{table}[h!]
\centering
\begin{tabular}{c c c c c c c}\hline \hline
 \textbf{ Min.}  &  \textbf{1st Qu.}  &  \textbf{Med.} &  \textbf{Mean}  &  \textbf{3rd Qu.}  &  \textbf{Max.} &  \textbf{Std.} \\
 3~726 & 43~930  & 96~480 &  115~760 &  144~100  & 561~200 &  102~408\\
\hline \hline
\end{tabular}
\caption{{\it Fluctuations of the number of observations $|\mathring{\mathcal{T}}_T|$ for $M = 100$ Monte-Carlo continuous trees observed up to time $T = 23$.
\label{tab : stat desc N_T}}}
\end{table}
We take a Gaussian kernel $K(x)=(2\pi)^{-1/2}\exp(-x^2/2)$ and the bandwidth $\widehat h_T$ is chosen here according to the rule-of-thumb  $1.06 \hat{\sigma}|\mathring{\mathcal{T}_T}|^{-1/5}$ where $\hat{\sigma}$ is the empirical standard deviation of $(\zeta_u,u\in\mathring{\mathcal{T}_T})$. We also implemented standard cross-validation with less success. We evaluate $\widehat{B}_T$ on a regular grid of $\mathcal{D} = [0.25,0.5]$ with mesh $\Delta x = 0.01$. For each sample we compute the empirical error 
\begin{equation*} \label{error def}
e_i = \frac{\|\widehat{B}^{(i)}_T - B\|_{\Delta x}}{\|B\|_{\Delta x}},\quad i=1,\ldots, M,
\end{equation*}
where $\|\cdot\|_{\Delta x}$ denotes the discrete norm over the numerical sampling. Table~\ref{tab : mean empirical error} displays the mean-empirical error $\overline{e}=M^{-1}\sum_{i=1}^M e_i$ together with the empirical standard deviation $\big( M^{-1}\sum_{i=1}^M (e_i-\overline{e})^2 \big)^{1/2}$.
\begin{table}[h!]
\begin{center}
\begin{tabular}{cccccccc}
\hline \hline
$\boldsymbol{T}$  &  $13$ & $15$ & $17$ & $19$ & $21$ & $23$
\\
\textbf{Mean} $\boldsymbol{|\mathring{\mathcal{T}}_T|}$  & 652 & 1~847 & 5~202 & 14~634 & 41~151 & 115~760  \\
$\boldsymbol{\overline{e}}$ & 0.1624 & 0.1046 & 0.0735 & 0.0448 & 0.0307 & 0.0178
\\
\textbf{Std. dev.} & \small{0.1052} & \small{0.0764} & \small{0.0599} & \small{0.0260} & \small{0.0197} & \small{0.0092} \\
\hline \hline
\end{tabular}
\end{center}
\caption{{\it Mean empirical relative error $\overline{e}$ and its standard deviation, with respect to $T$, for the division rate $B$ reconstructed over the interval $\mathcal{D} = [0.25,2.5]$ by the estimator $\widehat{B}_T$.} \label{tab : mean empirical error}
}
\end{table}
The comparison of the density of interest $f_B$ and the biased density $f_{H_B}$ on Figure~\ref{fig: IC} highlights the bias selection since $f_{H_B}$ gives more weight to small lifetimes than $f_B$. The error deteriorates as $x$ grows since the biased density $f_{H_B}$ (bold blue line - we approximate the Malthus parameter using \eqref{def malthus} and we find $\lambda_B \approx 0.5173$) decreases, see Figure~\ref{fig: IC}. 
The larger $T$, the better the reconstruction at a visual level, as shown on Figure~\ref{fig: IC} where $95\%$-level confidence  bands are built so that for each point $x$, the lower and upper bounds include $95\%$ of the estimators $(\widehat{B}_T^{(i)}(x), i = 1\ldots M)$. 
Close to $0$, $B(x)$ does not lie in the confidence band: our estimator exhibits a large bias there, and this is presumably due to a boundary effect. 
The error is close to $\exp(-2 \lambda_B T/5)$ as expected: indeed, for a kernel of order $n_0$, the bias term in density estimation is of order $h^{\beta\wedge(n_0+1)}$. Given that $B$ is smooth in our example, we rather expect $\exp(-\lambda_B \frac{(n_0+1)}{2(n_0+1)+1} T) = \exp(-2\lambda_BT/5) $ for the Gaussian kernel with $n_0=1$ that we use here, and this is consistent with what we observe in Figure~\ref{fig : erreur}.

\begin{figure}[h!]
\centering
\includegraphics[width=7cm]{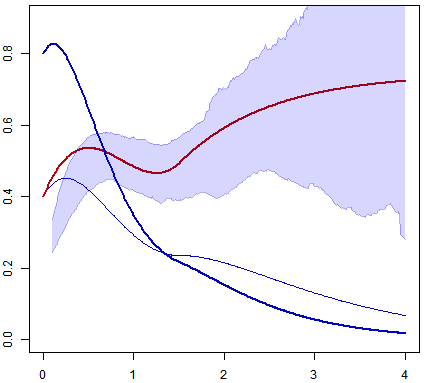}
\includegraphics[width=7cm]{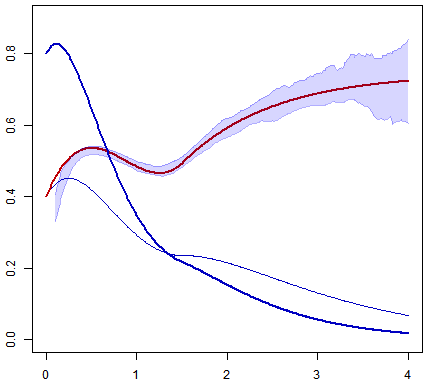}
\caption{{\it Reconstruction of $B$ over $\mathcal{D}= [0.1,4]$ with $95\%$-level confidence bands constructed over $M = 100$ Monte-Carlo continuous trees. In bold red line: $x \leadsto B(x)$; in bold blue line: $f_{H_B}$; in blue line: $f_B$ (on the same $y$-axis scale). Left: $T = 15$. Right: $T = 23$.} \label{fig: IC}}
\end{figure}

\begin{figure}[h!]
\centering
\includegraphics[width=6cm]{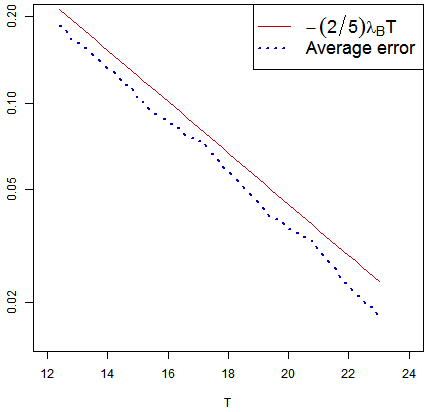}
\caption{{\it The log-average relative empirical error over $M = 100$ Monte-Carlo continuous trees vs. $T$ (i.e. the log-rate) for $x \leadsto B(x)$ reconstructed over $\mathcal{D} = [0.25,2.5]$ with $x \leadsto \widehat{B}_T(x)$ (dashed blue line) compared to the expected log-rate (solid red line). \label{fig : erreur}}}
\end{figure}


\section{\textsc{Proofs}} \label{sec: proofs}
For a locally integrable $B: [0,\infty) \rightarrow [0,\infty)$ such that $\int^\infty B(y)dy=\infty$, recall that we set 
$$f_B(x) = B(x)e^{-\int_0^x B(y)dy},\;\;x\geq 0.$$ 
Recall that $H_B$ is characterised by 
$$f_{H_B}(x) = me^{-\lambda_B x}f_{B}(x),\;\;x\geq 0.$$
\subsection{Preliminaries}
\subsubsection*{Many-to-one formulae}
For $u\in \mathcal U$, we write $\zeta_u^t$ for the age of the cell $u$ at time $t \in I_u=[b_u,d_u)$,  {\it i.e.} $\zeta_u^t = (t-b_u){\bf 1}_{\{t \in I_u\}}$. 
We extend $\zeta_u^t$ over $[0,b_u)$ by setting $\zeta_u^t = \zeta_{u(t)}^t$, where $u(t)$ is the ancestor of $u$ living at time $t$, defined by $u(t)=v$ if $v \preceq u$ and $(v,t) \in \mathbb T$. For $t\geq d_u$ we set $\zeta_u^t = \zeta_u$. Note that $\zeta_u^T=\zeta_u$ on the event $u \in \mathring{\mathcal{T}}_T$. 


\vip

Let $(\chi_t)_{t \geq 0}$ and $(\widetilde \chi_t)_{t \geq 0}$ denote the one-dimensional Markov processes with infinitesimal generators (densely defined on continuous functions vanishing at infinity) $\mathcal A_B$ and $\mathcal A_{H_B}$ respectively, where
$$\mathcal A_B g(x) = g'(x)+B(x)\big(g(0)-g(x)\big)$$
and such that $\PP(\chi_0=0)=\PP(\widetilde \chi_0=0)=1$. We also denote by $(P_{H_B}^t)_{t \geq 0}$ the Markov semigroup associated to $\mathcal A_{H_B}$.

\begin{prop}[Many-to-one formulae] \label{prop: Mt1}
For any $g \in \mathcal L_C$, we have 
\begin{equation} \label{Mt1 bord}
\E\Big[\sum_{u \in \partial {\mathcal T}_T}g(\zeta_u^T)\Big]   = 
\frac{e^{\lambda_BT}}{m}\E\Big[g(\widetilde \chi_T)B(\widetilde \chi_T)^{-1}H_B(\widetilde \chi_T)\Big],
\end{equation}
and
\begin{equation} \label{Mt1 int}
\E\Big[\sum_{u \in \mathring {\mathcal T}_T}g(\zeta_u^T)\Big]  = \E\Big[\sum_{u \in \mathring {\mathcal T}_T}g(\zeta_u)\Big]  = \frac{1}{m}\int_0^Te^{\lambda_B s}\E\Big[g(\widetilde \chi_s)H_B(\widetilde \chi_s)\Big]ds.
\end{equation}
\end{prop}
In order to compute rates of convergence, we will also need many-to-one formulae over pairs of individuals.  We can pick two individuals in the same lineage or over forks, {\it i.e.} over pairs of individuals that are not in the same lineage. If $u,v\in \mathcal U$, $u \wedge v$ denote their most recent common ancestor. Define
$$\mathcal{FU} = \{(u,v)\in \mathcal U^2, |u\wedge v| < |u| \wedge |v|\}\;\;\text{and}\;\;\mathcal{FT} = \mathcal {FU} \cap {\mathcal T}^2.$$
Introduce also $\bar m = \sum_{i\neq j}\sum_{k \geq i \vee j}p_k$ which is finite by Assumption~\ref{hyp: offspring}.
\begin{prop}[Many-to-one formulae over pairs] \label{prop: Mt1 forks}
For any $g \in \mathcal L_C$, we have 
\begin{align} \label{cloez forks}
\E\Big[\sum_{u,v \in \partial {\mathcal T}_T, \atop u \neq v}g(\zeta_u^T)g(\zeta_v^T)\Big] = \frac{\bar m }{m^3}\int_0^Te^{\lambda_Bs}\Big( e^{\lambda_B(T-s)} P^{T-s}_{H_B}\big(g\frac{H_B}{B}\big)(0)\Big)^2P^{s}_{H_B}H_B(0)ds,
\end{align}
\begin{align} \label{bansaye forks}
\E\Big[\sum_{(u,v)\in \mathcal{FT} \cap \mathring {\mathcal T}_T^2} g(\zeta_u)g(\zeta_v)\Big] = \frac{\bar m}{m^3} \int_0^T e^{\lambda_B s} \bigg( \int_0^{T-s} e^{\lambda_B t} P_{H_B}^{t}(g H_B)(0) dt \bigg)^2 P_{H_B}^s H_B(0)ds,
\end{align}
and 
\begin{align} \label{lineage}
\E\big[\sum_{u,v \in \mathring{\mathcal T}_T, \atop u \prec v}g(\zeta_u)g(\zeta_v)\big] = \int_0^T e^{\lambda_Bs} \Big( \int_0^{T-s}e^{\lambda_Bt} P^t_{H_B}\big(gH_B\big)(0)dt \Big) P^s_{H_B}(gH_B)(0)  ds.
\end{align}
\end{prop}
The identity \eqref{cloez forks} is a particular case of Lemma 3.9 of Cloez~\cite{cloez}. In order to obtain identity \eqref{bansaye forks}, we  closely follow the method of Bansaye {\it et al.}~\cite{BDMT}. Although the setting in~\cite{BDMT} is much more general than ours, it formally only applies for exponential renewal times (corresponding to constant functions $B$) so we need to slightly accommodate their proof. The same ideas enable us to prove \eqref{lineage}. This is set out in details in the appendix.
\subsubsection*{Geometric ergodicity of the auxiliary Markov process}
Define the probability measure 
$$\mu_B(x)dx=c_B \exp(-\int_0^x H_B(y)dy)dx\;\;\text{for}\;\;x\geq 0.$$
We have the fast convergence of $P_{H_B}^T$ toward $\mu_B$ as $T \rightarrow \infty$. More precisely, 
\begin{prop} \label{prop: convergence semigroupe}  Let $\rho_B = \inf_{x}H_B(x)$.
For any $B \in \mathcal B_{b,C}$, $g \in \mathcal{L}_{C'}$,  $t \geq 0$ and $x \in (0,\infty)$, we have
\begin{equation*} 
\Big|P_{H_B}^tg(x)-\int_0^\infty g(y)\mu_B(y)dy\Big| \leq 2 \sup_y|g(y)|\exp\big(-\rho_B t\big).
\end{equation*}
\end{prop}
\begin{proof}
First, one readily checks that $\int_0^{\infty} \mathcal A_{H_B}f(x)\mu_B(x) dx =0$ for any continuous $f$, and since moreover $P^t_{H_B}$ is Feller, it admits  $\mu_B(x)dx$ as an invariant probability. It is now sufficient to show 
$$
\|Q_B^{x, t}-\mu_B\|_{TV} \leq \exp(-\rho_B t)
$$
where $Q_B^{x,t}$ denotes the law of of the Markov process with infinitesimal generator $\mathcal A_{H_B}$ started from $x$ at time $t=0$ and $\|\cdot\|_{TV}$ is the total variation norm between probability measures. Let $N(ds \,dt)$ be a Poisson random measure with intensity $ds \otimes dt$ on $[0,\infty)\times [0,\infty)$. Define on the same probability space two random processes $(Y_t)_{t \geq 0}$ and $(Z_t)_{t \geq 0}$ such that
\begin{align*}
Y_t & = x+t-\int_0^t \int_0^\infty Y_{s_-}{\bf 1}_{\{z \leq H_B(Y_{s^-})\}}N(dz\, ds),\;\;t \geq 0, \\
Z_t & =  Z_0+ t -\int_0^t  \int_0^\infty Z_{s_-}{\bf 1}_{\{z \leq H_B(Z_{s^-})\}}N(dz \,ds),\;\;t \geq 0,
\end{align*}
where $Z_0$ is a random variable with distribution $\mu_B$. We have that both $(Y_t)_{t \geq 0}$ and $(Z_t)_{t \geq 0}$ are Markov processes driven by the same Poisson random measure with generator $\mathcal A_{H_B}$.  Moreover, if $N$ has a jump in $[0,t)\times[0, \inf_xH_B(x)]$, then $Y_t$ and $Z_t$ both necessarily start from $0$ after this jump and coincide further on. It follows that
$$\PP(Y_t \neq Z_t) \leq \PP\Big(\int_0^t \int_0^{\inf_x H_B(x)}N(ds\,dt)=0\Big) = \exp(-\inf_xH_B(x)t)=\exp(-\rho_Bt).$$
Observing that $Y_t$ and $Z_t$ have distribution $Q_B^{x,t}$ and $\mu_B$ respectively, we conclude thanks to the fact that $\|Q_B^{x, t}-\mu_B\|_{TV}\leq \PP(Y_t \neq Z_t)$.

\end{proof}

\subsection{Proof of Theorems~\ref{rate en T} and~\ref{rate avant T}}
In order to ease notation, when no confusion is possible, we abbreviate $\mathcal B_{b,C}$ by $\mathcal B$ and $\mathcal L_C$ by $\mathcal L$.
\begin{proof}[Proof of Theorem~\ref{rate en T}]
Writing
$$e^{\min\{\lambda_B/2,\rho_B\}T}\big(\mathcal E^T(\partial \mathcal{T}_T, g)-\partial \mathcal E_B(g)\big) = \frac{e^{\lambda_BT}}{\big|\partial{\mathcal T}_T\big|}\,e^{(\min\{\lambda_B/2,\rho_B\}-\lambda_B)T}\sum_{u \in \partial{\mathcal T}_T}\big(g(\zeta_u^T)-\partial \mathcal E_B(g)\big),$$
Theorem~\ref{rate en T} is then a consequence of the following two facts: first we claim that
\begin{equation} \label{conv proba}
e^{\lambda_BT}\big|\partial{\mathcal T}_T\big|^{-1} \rightarrow W_B\;\;\text{in probability as}\;\;T \rightarrow \infty, 
\end{equation}
uniformly in $B \in \mathcal B$, where the random variable $W_B$ satisfies  $\PP(W_B>0)=1$, and second, for $B \in \mathcal B$ and $g \in \mathcal L$, we claim that the following estimate holds:
\begin{equation} \label{second moment}
\E\Big[\Big(\sum_{u \in \partial{\mathcal T}_T}\big(g(\zeta_u^T)-\partial \mathcal E_B(g)\big)\Big)^2\Big] \lesssim e^{(2\lambda_B-\min\{\lambda_B, 2\rho_B\})T},
\end{equation}
where $\lesssim$ means up to a constant (possibly varying from line to line) that only depends on $\mathcal B$ and $\mathcal L$ and up to a multiplicative slow term of order $T$ in the case $\lambda_B=2\rho_B$.\\

\noindent {\it Step 1}. The convergence \eqref{conv proba} is a consequence of the following lemma:
\begin{lem} \label{estimate moment un}
For every $B \in \mathcal B$, there exists $\widetilde W_B$ with $\PP(\widetilde W_B>0)=1$ such that
\begin{equation} \label{convergence nombre moyen vivant en T}
\E\Big[\Big(\frac{|\partial{\mathcal T}_T|}{\E\big[|\partial{\mathcal T}_T|\big]}-\widetilde W_B\Big)^2\Big] \rightarrow 0\;\;\text{as}\;\;T\rightarrow \infty,
\end{equation}
uniformly in $B \in \mathcal B$ and
\begin{equation} \label{uniformity}
\kappa_B^{-1}e^{\lambda_BT}\E\big[|\partial{\mathcal T}_T|\big] \rightarrow 1\;\;\text{as}\;\;T\rightarrow \infty,
\end{equation}
uniformly in $B \in \mathcal B$, where $\kappa_B^{-1}=\lambda_B \tfrac{m}{m-1}\int_0^\infty\exp(-\int_0^xH_B(y)dy)dx$.
\end{lem}
Lemma~\ref{estimate moment un} is well known, and follows from classical renewal arguments, see Chapter 6 in the book of Harris~\cite{Harris}. Only the uniformity in $B \in \mathcal B$ requires an extra argument, but with a uniform version of the key renewal theorem of~\cite{unifkey}, it readily follows from the proof of Harris, so we omit it. Note that \eqref{convergence nombre moyen vivant en T} and \eqref{uniformity} entail the convergence $e^{\lambda_BT}|\partial{\mathcal T}_T|^{-1}\rightarrow \kappa_B\widetilde W_B^{-1}=W_B$ in probability as $T \rightarrow \infty$ uniformly in $B \in \mathcal B$, and this entails \eqref{conv proba}.\\

\noindent {\it Step 2}. We now turn to the proof of \eqref{second moment}. Without loss of generality, we may (and will) assume that $\partial \mathcal E_B(g)=0$. We have
$$\E\big[\big(\sum_{u \in \partial{\mathcal T}_T}g(\zeta_u^T)\big)^2\big] =  \E\big[\sum_{u \in \partial{\mathcal T}_T}g(\zeta_u^T)^2\big] + \E\big[\sum_{u, v \in \partial{\mathcal T}_T, \atop u\neq v}g(\zeta_u^T)g(\zeta_v^T)\big] = I + II,$$
say. By \eqref{Mt1 bord} in Proposition~\ref{prop: Mt1}, we write
\begin{align*}
I & = \frac{e^{\lambda_BT}}{m}\E\Big[g(\widetilde \chi_T)^2B(\widetilde \chi_T)^{-1}H_B(\widetilde \chi_T)\Big] \\
& \leq \frac{e^{\lambda_BT}}{m} \int_0^\infty g(x)^2\tfrac{H_B(x)}{B(x)}\mu_B(x)dx+\tfrac{e^{\lambda_BT}}{m}\Big|P^{T}_{H_B}\big(g^2\tfrac{H_B}{B}\big)(0)-\int_0^\infty g(x)^2\tfrac{H_B(x)}{B(x)}\mu_B(x)dx\Big|.
\end{align*}
Since $g \in \mathcal L$ and $B \in \mathcal B$, we successively have
$$m^{-1}\int_0^\infty g(x)^2\tfrac{H_B(x)}{B(x)}\mu_B(x)dx \lesssim 1\;\;\text{and}\;\;g(x)^2\tfrac{H_B(x)}{B(x)} \lesssim 1.$$
Note that for $B \in \mathcal B$, we have
$$H_B(x) = \frac{B(x)}{\int_x^\infty B(y)e^{-\lambda_B(y-x)}\exp(-\int_x^y B(u)du) dy}\leq \frac{b^2 \max\{C,1\}}{\lambda_B + b\max\{C,1\}}.$$
We also have $\lambda_B \leq  \lambda_{\widetilde B}$ as soon as $B(x) \leq \widetilde B(x)$ for all $x$ (see for instance the proof of Proposition \ref{classe B plus}) so $\inf_{B \in \mathcal B} \lambda_B>0$ and the uniformity in the above estimates follows likewise.
Applying Proposition~\ref{prop: convergence semigroupe} we derive 
$$\Big|P^{T}_{H_B}\big(g^2\tfrac{H_B}{B}\big)(0)-\int_0^\infty g(x)^2\tfrac{H_B(x)}{B(x)}\mu_B(x)dx\Big| \lesssim 1,$$
and we conclude that $I \lesssim e^{\lambda_BT} \leq e^{(2\lambda_B-\min\{\lambda_B, 2\rho_B\})T}$.
By \eqref{cloez forks} of Proposition~\ref{prop: Mt1 forks} we have
$$II  = \frac{\bar{m} e^{ 2\lambda_BT}}{m^3}\int_0^Te^{-\lambda_Bs} \Big(P^{T-s}_{H_B}\big(g\tfrac{H_B}{B}\big)(0)\Big)^2P^{s}_{H_B}H_B(0)ds.
$$
Since  $B \in \mathcal B$ and $g \in \mathcal L$, the estimates $P^{s}_{H_B}H_B(0) \lesssim 1$ and $|g(x)|\tfrac{H_B(x)}{B(x)} \lesssim 1$ hold true.  
Applying Proposition~\ref{prop: convergence semigroupe} to the test function $g(x)\tfrac{H_B(x)}{B(x)}$ which has vanishing integral under $\mu_B$, we obtain
$$\big|P^{T-s}_{H_B}\big(g\tfrac{H_B}{B}\big)(0)\big| \lesssim e^{-\rho_B(T-s)}$$ 
hence
\begin{align*}
|II| &\lesssim e^{2\lambda_BT}\int_0^T  e^{-\lambda_B s}e^{-2\rho_B(T-s)}ds \lesssim 
\left\{
\begin{array}{lll}
e^{\lambda_BT} & \text{if} & 2\rho_B \geq \lambda_B \\
e^{2(\lambda_B-\rho_B)T} & \text{if} & 2\rho_B < \lambda_B, 
\end{array}
\right.
\end{align*}
up to a multiplicative slow term of order $T$ when $2\rho_B = \lambda_B$. Note also that the estimate is uniform in $B \in \mathcal B$ since $\inf_{B \in \mathcal B}\lambda_B>0$ and $\inf_{B \in \mathcal B}\rho_B >0$. We conclude $|II| \lesssim e^{(2\lambda_B-\min\{\lambda_B, 2\rho_B\})T}$. 
\end{proof}
\begin{proof}[Proof of Theorem~\ref{rate avant T}]
The proof goes along the same line but is slightly more intricate. First, we implicitly work on the event $\{\big| \mathring {\mathcal T}_T\big|\geq 1\}$  which has probability that goes to $1$ as $T \rightarrow \infty$,  uniformly in $B \in \mathcal B$. We again write
$$e^{\min\{\lambda_B/2,\rho_B\}T}\big(\mathcal E(\mathring{\mathcal{T}_T}, g)-\mathring{\mathcal E}_B(g)\big) = \frac{e^{\lambda_BT}}{\big| \mathring {\mathcal T}_T\big|}e^{(\min\{\lambda_B/2,\rho_B\}-\lambda_B)T}\sum_{u \in  \mathring {\mathcal T}_T}\big(g(\zeta_u^T)-\mathring{\mathcal E}_B(g)\big),$$
and we claim that
\begin{equation} \label{convergence proba avant T}
e^{\lambda_BT}\big|\mathring{\mathcal T}_T\big|^{-1}\rightarrow W_B'>0\;\;\text{in probability as}\;\;T \rightarrow \infty, 
\end{equation}
uniformly in $B \in \mathcal B$, where $W_B'$ satisfies $\PP(W_B'>0)=1$ and that the following estimate holds:
\begin{equation} \label{second moment avant T}
\E\Big[\Big(\sum_{u \in \mathring{\mathcal T}_T}\big(g(\zeta_u^T)-\mathring{\mathcal E}_B(g)\big)\Big)^2\Big] \lesssim e^{(2\lambda_B-\min\{\lambda_B,2\rho_B\})T},
\end{equation}
uniformly in $B \in \mathcal B$ and $g \in \mathcal L$. 
In the same way as in the proof of Theorem~\ref{rate en T}, \eqref{convergence proba avant T} is a consequence of the following classical result, which can be obtained in the same way as for Lemma~\ref{estimate moment un} and proof which we omit.
\begin{lem} \label{estimate moment un avant T} 
For every $B \in \mathcal B$, there exists $\widetilde W_B'>0$ with $\PP(\widetilde W_B'>0)=1$ such that
$$
\E\Big[\Big(\frac{|\mathring{\mathcal T}_T|}{\E\big[|\mathring{\mathcal T}_T|\big]}-\widetilde W_B'\Big)^2\Big] \rightarrow 0\;\;\text{as}\;\;T\rightarrow \infty,
$$
uniformly in $B \in \mathcal B$
and
$$
(\kappa_{B}')^{-1}e^{\lambda_BT}\E\big[|\mathring{\mathcal T}_T|\big] \rightarrow 1\;\;\text{as}\;\;T\rightarrow \infty,
$$
uniformly in $B \in \mathcal B$, where $(\kappa_{B}')^{-1}=\lambda_B m \int_0^\infty\exp(-\int_0^xH_B(y)dy)dx$.
\end{lem}
It remains to prove \eqref{second moment avant T}. We again assume without loss of generality that $\mathring{\mathcal E}_B(g)=0$ and we plan to use the following decomposition:
\begin{equation} \label{eq:dec3}
\E\big[\big(\sum_{u \in \mathring{\mathcal T}_T}g(\zeta_u)\big)^2\big] = I + II + III,
\end{equation}
with
$$I = \E\big[\sum_{u \in \mathring{\mathcal T}_T}g(\zeta_u)^2\big],$$
$$II =\E\big[\sum_{(u, v) \in \mathcal{FT} \cap \mathring{\mathcal T}_T^2}g(\zeta_u)g(\zeta_v)\big] 
$$
and
$$III = 2\E\big[\sum_{u,v \in \mathring{\mathcal T}_T, \atop u \prec v}g(\zeta_u)g(\zeta_v)\big]. 
$$
\noindent {\it Step 1}. By \eqref{Mt1 int} of Proposition~\ref{prop: Mt1}, we have
$$I = \frac{1}{m}\int_0^Te^{\lambda_B s}\E\big[g(\widetilde \chi_s)^2H_B(\widetilde \chi_s)\big]ds,$$
In the same way as for the term $I$ in the proof of Theorem~\ref{rate en T}, we readily check that $g \in \mathcal L$ and $B\in \mathcal B$ guarantee that $ \E\big[g(\widetilde \chi_s)^2H_B(\widetilde \chi_s)\big] \lesssim 1$ therefore $I \lesssim e^{\lambda_BT} \leq e^{(2\lambda_B-\min\{\lambda_B,2\rho_B\})T}$.\\ 

\noindent {\it Step 2}. By \eqref{bansaye forks} of Proposition~\ref{prop: Mt1 forks}, we have 
\begin{align*}
II & = \frac{\bar m}{m^3}  \int_0^T e^{\lambda_Bs}\Big(\int_0^{T-s} e^{\lambda_B t} P_{H_B}^{t}(gH_B)(0)dt\Big)^2 P_{H_B}^s(H_B)(0)ds.
\end{align*}
We work as for the term $II$ in the proof of Theorem~\ref{rate en T}: we successively have $P_{H_B}^s(H_B)(0) \lesssim 1$ and
$\big|P_{H_B}^{t}(gH_B)(0)| \lesssim \exp(-\rho_Bt)$ by Proposition~\ref{prop: convergence semigroupe} and the fact that $gH_B$ has vanishing integral under $\mu_B$. Therefore
$$| II | \lesssim  \int_0^T e^{\lambda_Bs}\Big(\int_0^{T-s} e^{(\lambda_B-\rho_B)t}dt\Big)^2ds \lesssim 
\left\{
\begin{array}{lll}
e^{\lambda_BT} & \text{if} & 2\rho_B \geq \lambda_B \\
e^{2(\lambda_B-\rho_B)T} & \text{if} & 2\rho_B < \lambda_B 
\end{array}
\right.
$$
up to a multiplicative slow term of order $T$ when $2\rho_B = \lambda_B$.
We conclude $|II| \lesssim e^{(2\lambda_B-\min\{\lambda_B, 2\rho_B\})T}$ likewise.\\

\noindent{\it Step 3}. By \eqref{lineage} of Proposition~\ref{prop: Mt1 forks}, we have
\begin{align*}
| III | \leq \int_0^T e^{\lambda_Bs} \Big| \int_0^{T-s}e^{\lambda_Bt} P^t_{H_B}\big(gH_B\big)(0)dt \Big| P^s_{H_B}(|g|H_B)(0)  ds.
\end{align*}
In the same way as for the term $II$, we have  $|P_{H_B}^{s}(|g|H_B)(0)| \lesssim 1$ and $|P_{H_B}^{t}(gH_B)(0)| \lesssim \exp(-\rho_Bt)$. Therefore
$$|III| \lesssim \int_0^T e^{\lambda_Bs} \big(\int_0^{T-s}e^{(\lambda_B-\rho_B)t}dt\big)ds\lesssim e^{\lambda_BT} \leq e^{(2\lambda_B-\min\{\lambda_B, 2\rho_B\})T}.$$
\end{proof}
\subsection{Proof of Proposition~\ref{prop: conv lambda}}
Conditional on $\mathring{\mathcal T}_T$, the random variables $(\nu_u, u \in \mathring{\mathcal T}_T)$ are independent, with common distribution $p_k$. It follows that 
$$\E\big[(\widehat m_T - m)^2\,|\mathring{\mathcal T}_T\big] \leq |\mathring{\mathcal T}_T|^{-1}\sum_kk^2p_k.$$
Since $e^{\lambda_BT}|\mathring{\mathcal T}_T|^{-1}$ is $\mathcal B$-tight thanks to Lemma ~\ref{estimate moment un avant T}, we obtain the result for $e^{\lambda_BT/2}(\widehat m_T - m)$.  
The $\mathcal B$-tightness of $T^{-1} v_T(B)^{-1}(\widehat \lambda_T - \lambda_B)$ is a consequence of Theorem~\ref{rate en T} and~\ref{rate avant T}, together with the convergence  of the preliminary estimators $\widehat m_T$. For $M>0$,
\begin{multline*}
\mathcal{E}^T(\mathrm{Id},\partial \mathcal{T}_T) - \partial \mathcal{E}_B (\mathrm{Id}) 
= \big( \mathcal{E}^T(\min\{\mathrm{Id}, M\},\partial \mathcal{T}_T) - \partial \mathcal{E}_B (\min\{\mathrm{Id}, M\}) \big) \\+ \Big( \mathcal{E}^T \big((\cdot - M) {\bf 1}_{\{ \cdot > M \}},\partial \mathcal{T}_T\big) - \partial \mathcal{E}_B \big((\cdot - M) {\bf 1}_{\{ \cdot > M \}}\big) \Big) = I + II
\end{multline*}
say. We choose $M=M_T = 2T$ and we apply Theorem~\ref{rate en T} for the test functions $g_T(x) = \min\{x, M_T\}/M_T$ which are uniformly bounded in $T$ to get the $\mathcal{B}$-tightness of $T^{-1} v_T(B)^{-1} I$. Since $\zeta_u \leq T$ when $u \in \partial T_T$, we also have
$| II | = \partial \mathcal{E}_B \big((\cdot - 2T) {\bf 1}_{\{ \cdot > 2T \}}\big)$ and we deduce that $v_T(B)^{-1} II$ is $\mathcal{B}$-tight. We study in the same way $\mathcal{E}^T(\mathrm{Id},\mathring{\mathcal{T}}_T)$ to conclude.

\subsection{Proof of Theorem~\ref{thm: upper rate}}
The proof of Theorem~\ref{thm: upper rate} goes along the classical line of a bias-variance analysis in nonparametrics (see for instance the classical textbook~\cite{Tsybakov}). However, we have two kind of extra difficulties: first we have to get rid of the random bandwidth $\widehat h_T=\exp(-  \widehat \lambda_T  \tfrac{1}{2\beta+1} T )$ defined in \eqref{def bandwidth} (actually the most delicate part of the proof) and second, we have to get rid of the preliminary estimators $\widehat m_T$ and $\widehat \lambda_T$.\\

The point $x \in (0,\infty)$ where we estimate $B(x)$ is fixed throughout, and further omitted in the notation. We first need a slight extension of Theorem~\ref{rate avant T} -- actually of the estimate \eqref{second moment avant T} -- in order to accommodate test functions $g = g_T$ such that $g_T \rightarrow \delta_x$ weakly as $T\rightarrow \infty$.
For a function $g:[0,\infty)\rightarrow \R$ let
$$|g|_{1} = \int_0^\infty |g(y)|dy,\;\;|g|_{2}^2=\int_0^\infty g(y)^2dy\;\;\text{and}\;\;|g|_{\infty}=\sup_{y}|g(y)|$$ 
denote the usual $L^p$-norms over $[0,\infty)$ for $p=1,2, \infty$. Define also
\begin{equation}  \label{def phi}
\Phi_T(B,g) =
\left\{
\begin{array}{lll
} |g|_2^2 + \inf_{0 \leq v \leq T} \big(|g|_1^2e^{\lambda_Bv} + |g|_\infty^2e^{(2(\lambda_B-\rho_B)_+-\lambda_B)v}\big)+|g|_1|g|_\infty & \text{if} & \lambda_B \leq 2\rho_B \\ \\
|g|_2^2+|g|_\infty^2e^{(\lambda_B-2\rho_B)T}+|g|_1|g|_\infty  & \text{if} & \lambda_B > 2\rho_B.
\end{array}
\right.
\end{equation}
\begin{prop} \label{lem: vanishing test}
In the same setting as Theorem~\ref{rate avant T}, we have, for any $g \in \mathcal L$,
\begin{align} \label{gagner du h}
& \E\Big[\Big(\sum_{u \in \mathring{\mathcal T}_T}\big(g(\zeta_u^T)-\mathring{\mathcal E}_B(g)\big)\Big)^2\Big] 
\lesssim   e^{(\lambda_B-\rho_B)_+T}|g|_\infty^2+e^{\lambda_BT} \Phi_T\big(B,g\big),
\end{align}
where the symbol $\lesssim$ means here uniformly in $B \in \mathcal B$ and independently of $g$. 
\end{prop}
Let us briefly comment on Proposition~\ref{lem: vanishing test}. If $g$ is bounded and compactly supported with $\int g=1$, consider the function $g_{h_T}(y) = h_T^{-1}g\big(h_T^{-1}(x-y)\big)$  that mimics the Dirac function $\delta_x$ for $h_T \rightarrow 0$. It is noteworthy that in the left-hand side of \eqref{gagner du h}, $g_{h_T}(\zeta_u^T)^2$ is of order $h_T^{-2}$ while the right-hand side is of order $e^{\lambda_BT}h_T^{-1}$ if we pick $\omega = h_T^{-1}$ (allowed as soon as $h_T^{-1} \leq e^{\lambda_{B}T}$). We can thus expect to gain a crucial factor $h_T$ thanks to averaging over $\mathring{\mathcal T}_T$.
\begin{proof}
We carefully revisit the estimate \eqref{second moment avant T} in the proof of Theorem~\ref{rate avant T} keeping up with the same notation and assuming with no loss of generality that $\mathring{\mathcal E}_B(g)=0$. Recall decomposition \eqref{eq:dec3}.
\vip
\noindent {\it Step 1}. 
For the term $I$, we insert  $\int_0^\infty g(y)^2H_B(y)\mu_B(y)dy = mc_B\int_0^\infty g(y)^2 e^{-\lambda_B y} f_B(y)dy$ to obtain 
$I=IV+V,$
where 
$$IV \lesssim e^{\lambda_BT}\int_0^\infty  g(y)^2 e^{-\lambda_B y} f_B(y)dy$$
and 
\begin{align*}
|V| & \leq \frac{1}{m}\int_0^T e^{\lambda_Bs}\Big|P_{H_B}^s\big(g^2 H_B\big)(0)-\int_0^\infty g(y)^2H_B(y)\mu_B(y)dy\Big|ds.
\end{align*}
Clearly, $|IV| \lesssim e^{\lambda_BT}|g|_{2}^2$. By Proposition~\ref{prop: convergence semigroupe}, we further infer
$$|V| \lesssim |g|_\infty^2\int_0^T e^{\lambda_Bs}e^{-\rho_Bs}ds \lesssim |g|_\infty^2e^{(\lambda_B-\rho_B)_+T}.$$
\vip
\noindent {\it Step 2}. For the term $II$, 
using $P_{H_B}^s(H_B)(0) \lesssim 1$ we now obtain
\begin{align*} 
II 
& \lesssim e^{\lambda_BT}\int_0^T e^{-\lambda_B s} \Big(\int_0^s e^{\lambda_B t} P_{H_B}^{t}(g H_B)(0) dt\Big)^2ds.
\end{align*}
A new difficulty appears here, since the crude bound 
\begin{equation} \label{crude ergo}
| P_{H_B}^{t}(gH_B)(0)| \lesssim |g|_\infty \exp(-\rho_B t)
\end{equation}
given by Proposition~\ref{prop: convergence semigroupe} does not yield to the correct order for small value of $t$ because of the term $|g|_\infty$.  We need the following refinement (for small values of $t$), based on a renewal argument and proved in Appendix:
\begin{lem} \label{lem: renouvellement}
For every $t\geq 0$ and $g\in \mathcal L$, we have
$$\big|P^t_{H_B}\big(gH_B\big)(0)\big| \lesssim 
|g(t)|e^{-\lambda_B t}+ |g|_1$$
uniformly in $B \in \mathcal B$.
\end{lem}
Let $v \in [0,T]$ be arbitrary. For $0 \leq s  \leq v$, by Lemma~\ref{lem: renouvellement} we obtain
\begin{equation*}
\mathcal I_s =\Big(\int_0^s e^{\lambda_B t} |P_{H_B}^{t}(gH_B)(0)| dt\Big)^2  \lesssim \Big( \int_0^s |g(t)| dt+|g|_1\int_0^s e^{\lambda_B t} dt\Big)^2  \lesssim |g|_1^2 e^{2\lambda_B s}.
\end{equation*}
For $s \geq v$, we have by \eqref{crude ergo}  
\begin{align*}
\mathcal I_s & \lesssim  \mathcal I_{v} + |g|_\infty^2\Big(\int_v^se^{(\lambda_B-\rho_B)t}dt\Big)^2 \lesssim \mathcal I_{v} + |g|_\infty^2\big(e^{2(\lambda_B-\rho_B)s}{\bf 1}_{\{\lambda_B > \rho_B\}}+(s- v)^2 {\bf 1}_{\{\lambda_B \leq \rho_B\}}{\bf 1}_{\{s \geq v\}}\big).
\end{align*}
On the one hand, 
$\int_0^{v} e^{-\lambda_B s}\mathcal I_sds  \lesssim  |g|_1^2 e^{\lambda_B v}$
and on the other hand $\int_{v}^Te^{-\lambda_Bs}\mathcal I_sds$ is less than
\begin{align*}
\; \mathcal I_{v}&\int_{v}^T e^{-\lambda_Bs}ds+ |g|_\infty^2\Big(\int_v^T e^{-\lambda_Bs}e^{2(\lambda_B-\rho_B)_+s}ds+\int_v^Te^{-\lambda_Bs}(s-v)^2ds{\bf 1}_{\{\lambda_B \leq \rho_B\}}\Big) \\  
\lesssim &
\left\{
\begin{array}{lll}
|g|_1^2e^{\lambda_Bv}+|g|_\infty^2e^{-\lambda_B v} & \text{if} & \lambda_B \leq \rho_B \\ \\
|g|_1^2e^{\lambda_Bv}+|g|_\infty^2e^{(\lambda_B-2\rho_B)v} & \text{if} & \rho_B \leq \lambda_B \leq 2\rho_B \\ \\
|g|_1^2e^{\lambda_Bv}+|g|_\infty^2e^{(\lambda_B-2\rho_B)T} & \text{if} & \lambda_B \geq 2\rho_B, \\
\end{array}
\right.
\end{align*}
whence for every $v \in [0,T]$, we derive
$$|II| \lesssim e^{\lambda_B T}\Big(|g|_1^2e^{\lambda_Bv} + |g|_\infty^2\big(e^{(-\lambda_B+2(\lambda_B-\rho_B)_+)v}{\bf 1}_{\{\lambda_B \leq 2\rho_B\}}+ e^{(\lambda_B-2\rho_B)T}{\bf 1}_{\{\lambda_B > 2\rho_B\}}\big)
\Big).$$
\vip
\noindent {\it Step 3}. 
Finally going back to Step 3 in the proof of Theorem~\ref{rate avant T} we readily obtain
\begin{align*}
|III|& \lesssim  \int_0^T e^{\lambda_B s} P_{H_B}^s\big(|g|H_B\big)(0)\int_0^{T-s}e^{\lambda_B t}\big|P_{H_B}^t\big(gH_B\big)(0)\big|dtds \\
& \lesssim \int_0^T e^{\lambda_Bs} (|g(s)|e^{-\lambda_Bs}+|g|_1)|g|_\infty\int_0^{T-s} e^{\lambda_Bt}e^{-\rho_Bt}dtds  
\end{align*}
by applying Lemma~\ref{lem: renouvellement} for the term involving $P^s_{H_B}$ and the estimate \eqref{crude ergo} for the term involving $P_{H_B}^t$, therefore $|III| \lesssim  e^{\lambda_B T}|g|_1|g|_\infty$. 
\end{proof}
\vip
Proposition~\ref{lem: vanishing test} enables us to obtain the next result which is the key ingredient to get rid of the random bandwidth $\widehat h_T$, thanks to the fact that it is concentrated around its estimated value $h_T(\beta)= e^{- \frac{1}{2\beta+1}\lambda_BT}$. To that end, define, for $C >0$
$${\mathcal C}_C = \big\{g: \R \rightarrow \R,\;\mathrm{supp}(g) \subset [0,C] \;\text{and}\;\sup_y|g(y)|\leq C \big\}.$$
Denote by $\mathcal C_C^1$ (later abbreviated by $\mathcal C^1$) the subset of $\mathcal C_C$ of functions that are moreover differentiable, with derivative uniformly bounded by $C$. For $h>0$ we set $g_h(y) = h^{-1} g\big(h^{-1}(x-y)\big)$. Finally, for $a,b\geq 0$ we set  $[a\pm b]=[(a-b)_+, a+b]$. Recall from Section \ref{convergence statistique} that 
$v_T(B) = e^{-\min\{\rho_B,\lambda_B/2\} T}$ if $\lambda_B \neq 2\rho_B$ and $T^{1/2}e^{-\lambda_BT/2}$ otherwise.
\begin{prop} \label{lem: kolmo} 
Assume that $\beta \geq 1/2$. Define $\varpi_B = \min\{\max\{1,\lambda_B/\rho_B\}, 2\}$.
 For every $\kappa>0$, 
$$v_T(B)^{-1}\sup_{h \in [h_T(\beta)(1\pm \kappa T^2 v_T(B))]} \big|\mathcal E^T\big(\mathring{\mathcal T}_T,h^{\varpi_B /2}fg_{h}) - \mathring{\mathcal E}_B(h^{\varpi_B /2}fg_{h})\big|$$
is $\mathcal B \times \mathcal L \times \mathcal C^1$-tight for the parameter $(B,f,g)$.
\end{prop}
\begin{proof}
\noindent {\it Step 1}. Define  $\overline{fg_h} = fg_h-\mathring{\mathcal E}_B(fg_{h})$. Writing
\begin{align*}
&v_T(B)^{-1}\Big(\mathcal E^T\big(\mathring{\mathcal T}_T,h^{\varpi_B /2}fg_{h}) - \mathring{\mathcal E}_B(h^{\varpi_B /2}fg_{h})\Big) \\
= &\frac{e^{\lambda_BT}}{|\mathring{\mathcal T}_T|} e^{(\min\{\rho_B,\lambda_B/2\}-\lambda_B)T}(T^{-1/2})^{{\bf 1}_{\{\lambda_B=2\rho_B\}}}\sum_{u \in \mathring{\mathcal T}_T}h^{\varpi_B/2}\overline{fg_h}(\zeta_u),
\end{align*}
we see as in the proof of Theorem~\ref{rate avant T}  that thanks to Lemma~\ref{estimate moment un avant T},  it is enough to prove the $\mathcal B$-tightness of 
$$\sup_{h \in[h_T(\beta)(1\pm \kappa T^2 v_T(B))] }|V^T_h| =  \sup_{s \in [0,1]}|V^T_{h_s}|,$$ where
$$V^T_h = e^{(\min\{\rho_B,\lambda_B/2\}-\lambda_B)T}(T^{-1/2})^{{\bf 1}_{\{\lambda_B=2\rho_B\}}}\sum_{u \in \mathring{\mathcal T}_T}h^{\varpi_B /2}\overline{fg_h}(\zeta_u),$$
and
$$h_s = h_T(\beta)\big(1-\kappa T^2 v_T(B)\big)+2s\kappa h_T(\beta)T^2 v_T(B), \;\;s \in [0,1].$$
\noindent {\it Step 2}. We claim that 
\begin{equation} \label{kolmo a prouver}
\left\{
\begin{array}{ll}
\sup_{T>0}\E\big[(V^T_{h_0})^2\big] <\infty &  \\ \\
 \E\big[\big(V^T_{h_t}-V^T_{h_s}\big)^2\big]  \leq C'(t-s)^2 & \text{for} \; s, t\in [0,1], \\ 
 \end{array}
\right.
\end{equation} 
for some constant $C'>0$ that does not depend on $T$ nor $B \in \mathcal B$. Then, by Kolmogorov continuity criterion, this implies in particular that 
$$\sup_{T >0}\sup_{B \in \mathcal B}\E\big[\sup_{s \in [0,1]}|V^T_{h_s}|\big] <\infty$$ hence the result (see for instance~\cite{kolmo unif} to track the constant and obtain a uniform version of the continuity criterion). We have
\begin{align*}
V^T_{h_t}-V^T_{h_s} 
& = e^{(\min\{\rho_B,\tfrac{\lambda_B}{2}\}-\lambda_B)T}(T^{-\tfrac{1}{2}{\bf 1}_{\{\lambda_B=2\rho_B\}}})\sum_{u \in \mathring{\mathcal T}_T}\Big(\Delta_{s,t}(h^{\varpi_B /2}fg_{h})(\zeta_u) - \mathring{\mathcal E}_B\big(\Delta_{s,t}(h^{\varpi_B /2}fg_{h})\big)\Big)
\end{align*}
where
$\Delta_{s,t}(h^{\varpi_B /2}fg_{h})(y) = h_t^{\varpi_B /2}f(y)g_{h_t}(y)-h_s^{\varpi_B /2}f(y)g_{h_s}(y).$
By Proposition~\ref{lem: vanishing test}, we derive that $\E[(V^T_{h_t}-V^T_{h_s})^2]$ is less than \\
%
$$
\left\{
\begin{array}{lll}
e^{-\lambda_BT}|\Delta_{s,t}(h^{\varpi_B/2}fg_{h})|_\infty^2+\Phi_T \big(B,\Delta_{s,t}(h^{\varpi_B /2}fg_{h})\big) & \text{if} & \lambda_B \leq \rho_B \\ \\
e^{-\rho_BT}|\Delta_{s,t}(h^{\varpi_B /2}fg_{h})|_\infty^2+\Phi_T \big(B,\Delta_{s,t}(h^{\varpi_B /2}fg_{h})\big) & \text{if} & \rho_B \leq \lambda_B \leq 2 \rho_B \\ \\
e^{-(\lambda_B-\rho_B)T}|\Delta_{s,t}(h^{\varpi_B /2}fg_{h})|_\infty^2+e^{-(\lambda_B-2\rho_B)T}\Phi_T \big(B,\Delta_{s,t}(h^{\varpi_B /2}fg_{h})\big) & \text{if} & \lambda_B \geq 2 \rho_B \\ 
\end{array}
\right.
$$
(we ignore the slow term in the limiting case $\lambda_B = 2\rho_B$) and the remainder of the proof amounts to check that each term in the estimate above has order $(t-s)^2$ uniformly in $T$ and $B \in \mathcal B$.
\vip

\noindent {\it Step 3}.  For every $y$, we have
$$\Delta_{s,t}\big(h^{\varpi_B /2}fg_{h}\big)(y) = (h_t-h_s)\partial_h\big(h^{\varpi_B /2}f(y)g_{h}(y)\big)_{|h=h^\star(y)}$$ 
for some $h^\star(y) \in [\min\{h_t,h_s\},\max\{h_t,h_s\}]$. 
Observe now that since $g \in \mathcal C^1$ and $f \in \mathcal L$, we have
$$
\partial_h\big( h^{\varpi_B /2}fg_{h}(y)\big) = (\tfrac{\varpi_B}{2}-1 )h^{\tfrac{\varpi_B}{2}-2}f(y)g\big(h^{-1}(x-y)\big)-h^{\tfrac{\varpi_B}{2}-3}(x-y)f(y)g'\big(h^{-1}(x-y)\big)
$$
therefore, for small enough $h$ (which is always the case for $T$ large enough, uniformly in $B \in \mathcal B$) and since $|x-y| \lesssim h$ thanks to the fact that $g$ is compactly supported, we obtain
\begin{equation*} 
|\partial_h \big(h^{\varpi_B/2}fg_{h}(y)\big)| \lesssim h^{\varpi_B/2-2}{\bf 1}_{[0,C]}\big(h^{-1}(x-y)\big).
\end{equation*}
Assume with no loss of generality that $s\leq t$ so that $h_s \leq h(y)^\star \leq h_t$. It follows that 
\begin{align*}
\big|\Delta_{s,t}\big(h^{\varpi_B/2}fg_{h})(y)\big| & \lesssim (h_t-h_s)h^\star(y)^{\varpi_B/2-2}{\bf 1}_{[0,C]}\big(h^\star(y)^{-1}(x-y)\big) \\
& \leq  (h_t-h_s)h_s^{\varpi_B/2-2}{\bf 1}_{[0,C]}\big(h_t^{-1}(x-y)\big). 
\end{align*}
Using that $h_t-h_s = 2(t-s)\kappa T^2h_T(\beta)v_T(B)$, we successively obtain 
\begin{align*}
& \big|\Delta_{s,t}(h^{\varpi_B/2}fg_{h})\big|_\infty^2 \lesssim (h_t-h_s)^2h_s^{\varpi_B-4}  \lesssim (t-s)^2 T^4 v_T(B)^2h_T(\beta)^{\varpi_B-2}, \\ 
& \big|\Delta_{s,t}(h^{\varpi_B/2}fg_{h})\big|_2^2  \lesssim (h_t-h_s)^2h_s^{\varpi_B-4}h_t \lesssim (t-s)^2 T^4 v_T(B)^2h_T(\beta)^{\varpi_B-1},\\ 
& \big|\Delta_{s,t}(h^{\varpi_B/2}fg_{h})\big|_1^2 \lesssim (h_t-h_s)^2h_s^{\varpi_B-4}h_t^2\lesssim (t-s)^2 T^4 v_T(B)^2h_T(\beta)^{\varpi_B},\\ 
 & \big|\Delta_{s,t}(h^{\varpi_B/2}fg_{h})\big|_1\big|\Delta_{s,t}\big(h^{\varpi_B/2}fg_{h})\big|_\infty  
 \lesssim (t-s)^2 T^4 v_T(B)^2h_T(\beta)^{\varpi_B-1}.
\end{align*}
\vip

\noindent {\it Step 4}. Recall that $h_T(\beta)=e^{-\lambda_BT/(2\beta+1)}$. When $\lambda_B \leq \rho_B$, we have $v_T(B)=e^{-\lambda_BT/2}$ and $\varpi_B=1$. By definition  of $\Phi_T$ in \eqref{def phi} together with the estimates of Steps 2 and 3, we obtain
\begin{align*}
\E\big[(V^T_{h_t}-V^T_{h_s})^2\big] 
& \lesssim e^{-\lambda_BT}|\Delta_{s,t}(h^{1/2}fg_{h})|_\infty^2+\Phi_T \big(B,\Delta_{s,t}(h^{1/2}fg_{h})\big) \\
& \lesssim (t-s)^2T^4\big(e^{\lambda_B(\tfrac{1}{2\beta+1}-2)T}+e^{-\lambda_BT}+e^{-\lambda_B(\tfrac{1}{2\beta+1}+1)T}e^{\lambda_B v}+e^{\lambda_B(\tfrac{1}{2\beta+1}-1)T}e^{-\lambda_Bv}\big) 
\end{align*}
which is of order $(t-s)^2$ uniformly in $T>0$ by picking $v=0$ for instance. When $\rho_B \leq \lambda_B \leq 2\rho_B$, we still have $v_T(B)=e^{-\lambda_BT/2}$ but now $\varpi_B= \lambda_B/\rho_B$. It follows that $\E[(V^T_{h_t}-V^T_{h_s})^2]$ is of order
\begin{align*}
 & e^{-\rho_BT}|\Delta_{s,t}(h^{\lambda_B/2\rho_B}fg_{h})|_\infty^2+\Phi_T \big(B,\Delta_{s,t}(h^{\lambda_B/2\rho_B}fg_{h})\big) \\
 \lesssim &\; (t-s)^2T^4\big(e^{\lambda_B(\tfrac{2-\lambda_B/\rho_B}{2\beta+1}-1)T}(e^{-\rho_BT}+e^{(\lambda_B-2\rho_B)v})+e^{\lambda_B(\tfrac{1-\lambda_B/\rho_B}{2\beta+1}-1)T}+e^{-\lambda_B(\tfrac{\lambda_B/\rho_B}{2\beta+1}+1)T}e^{\lambda_Bv}\big) 
\end{align*}
and this last term is again of order $(t-s)^2$ uniformly in $T>0$ by noting that $1 \leq \lambda_B/\rho_B\leq 2$ and picking $v=0$ for instance. Finally, when $2\rho_B \leq \lambda_B$, we have $v_T(B)=e^{-\rho_BT}$ and $\varpi_B=2$. This entails 
\begin{align*}
\E\big[(V^T_{h_t}-V^T_{h_s})^2\big] 
& \lesssim e^{-(\lambda_B-\rho_B)T}|\Delta_{s,t}(hfg_{h})|_\infty^2+e^{-(\lambda_B-2\rho_B)T}\Phi_T \big(B,\Delta_{s,t}(hfg_{h})\big) \\
& \lesssim (t-s)^2T^4\big(e^{-(\lambda_B+\rho_B)T}+e^{-\lambda_B(\tfrac{1}{2\beta+1}+1)T}+e^{-2\rho_BT}\big) 
\end{align*}
and these terms are all again of order $(t-s)^2$ uniformly in $T$. 
\vip

\noindent {\it Step 5}. It remains to show $\sup_{T>0}\E\big[(V^T_{h_0})^2\big] <\infty$ in order to complete the proof of 
\eqref{kolmo a prouver}. By Step 2 and the definition of $\varpi_B$, we readily have 
$$\E\big[(V_{h_0}^T)^2\big] \lesssim
\left\{
\begin{array}{lll}
e^{-\lambda_BT}|h_0^{1/2}fg_{h_0}|_\infty^2+\Phi_T(B,h_0^{1/2}fg_{h_0}) & \text{if} & \lambda_B \leq \rho_B\\ \\
e^{-\rho_BT}|h_0^{\lambda_B/2\rho_B}fg_{h_0}|_\infty^2+\Phi_T (B,h_0^{\lambda_B/2\rho_B}fg_{h_0}) & \text{if} & \rho_B \leq \lambda_B \leq 2 \rho_B \\ \\
e^{-(\lambda_B-\rho_B)T}|h_0fg_{h_0}|_\infty^2+e^{-(\lambda_B-2\rho_B)T}\Phi_T(B,h_0fg_{h_0}) & \text{if} & \lambda_B \geq 2 \rho_B. \\ 
\end{array}
\right.
$$
When $\lambda_B \leq \rho_B $, since $h_0$ is of order $h_T(\beta)$, we have
$$\E\big[(V_{h_0}^T)^2\big] \lesssim e^{-\lambda_BT}h_T(\beta)^{-1}+1+h_T(\beta)e^{\lambda_Bv}+h_{T}(\beta)^{-1}e^{-\lambda_B v}$$
for every $v \in [0,T]$, and the choice $v=\frac{1}{2\beta+1}T$ entails $\E[(V_{h_0}^T)^2] \lesssim 1$. When $\rho_B \leq \lambda_B \leq 2\rho_B$, we have
$$\E\big[(V_{h_0}^T)^2\big] \lesssim e^{-\rho_BT}h_T(\beta)^{\tfrac{\lambda_B}{\rho_B}-2}+h_T(\beta)^{\tfrac{\lambda_B}{\rho_B}-1}+h_T(\beta)^{\tfrac{\lambda_B}{\rho_B}}e^{\lambda_Bv}+h_{T}(\beta)^{\tfrac{\lambda_B}{\rho_B}-2}e^{(\lambda_B-2\rho_B)v}.$$
The first term is bounded as soon as $\beta \geq 1/2$ and the choice $v=\tfrac{\lambda_B}{\rho_B(2\beta+1)}T$ for the last two terms entails $\E[(V_{h_0}^T)^2] \lesssim 1$. Finally, when $2\rho_B \leq \lambda_B $ we have
$$\E\big[(V_{h_0}^T)^2\big] \lesssim e^{-(\lambda_B-\rho_B)T}+1$$
and this term is bounded likewise. Eventually \eqref{kolmo a prouver} is established and Proposition~\ref{lem: kolmo} is proved. 
\end{proof}
We now get rid of the preliminary estimators $\widehat m_T$ and $\widehat \lambda_T$. Remember that the target rate of convergence for $\widehat B_T(x)$ is $w_T(B)=T^{{\bf 1}_{\{\lambda_B=2\rho_B\}}} \exp\big(-\min\{\lambda_B, 2\rho_B\}\frac{\beta-(\lambda_B/\rho_B-1)_+/2}{2\beta+1}T\big)$.  
\begin{lem} \label{lem: degager est preli}
Assume that $\beta>1$. Let either $G_T(y)=g_{\widehat h_T}(y)$ with $g \in {\mathcal C}^1$ or $G_T(y) = {\bf 1}_{\{y \leq x\}}$ for $y \in [0,\infty)$. Then
\begin{equation*}
w_T(B)^{-1}\big(\mathcal E^T(\mathring{\mathcal T}_T, \widehat m_T^{-1}e^{\widehat\lambda_T\cdot}G_T)-\mathcal E^T\big(\mathring{\mathcal T}_T, m^{-1}e^{\lambda_B\cdot}G_T)\big)
\end{equation*}
is $\mathcal B$-tight for the parameter $B$.
\end{lem}
\begin{proof}
For $u \in \mathring{\mathcal T}_T$ and its lifetime $\zeta_u$, define
$$\gamma_T(u)=w_T(B)^{-1}\big(\widehat m_T^{-1}e^{\widehat\lambda_T\zeta_u}-m^{-1}e^{\lambda_B\zeta_u}\big)G_{T}(\zeta_u).$$
Lemma~\ref{lem: degager est preli} amounts to show that $|\mathring{\mathcal T}_T|^{-1} \sum_{u \in \mathring{\mathcal T}_T}\gamma_T(u)$ is $\mathcal B$-tight. Set $h_T(\beta) = \exp(-\lambda_B \tfrac{1}{2\beta+1}T)$ and note that
$$w_T(B)^{-1}=(T^{-1/2})^{{\bf 1}_{\{\lambda_B=2\rho_B\}}}e^{\min\{\rho_B,\lambda_B/2\}T}h_T(\beta)^{\varpi_B/2}=v_T(B)^{-1}h_T(\beta)^{\varpi_B/2},$$
where $\varpi_B = \min\{\max\{1,\lambda_B/\rho_B\}, 2\}$. We first treat the case $G_T(y)=g_{\widehat h_T}(y)$.\\

\noindent {\it Step 1.} By Proposition~\ref{prop: conv lambda}, we have
$$\widehat \lambda_T = \lambda_B + Tv_T(B)r_T\;\;\text{and}\;\;\widehat m_T^{-1} = m^{-1} + e^{-\lambda_BT/2} r'_T,$$
where both $r_T$ and $r'_T$ are $\mathcal B$-tight. We then have the decomposition
\begin{align*}
\gamma_T(u) & = w_T(B)^{-1}\widehat m_T^{-1}(e^{\widehat \lambda_T\zeta_u}-e^{\lambda_B \zeta_u})g_{\widehat h_T}(\zeta_u)+w_T(B)^{-1}(\widehat m^{-1}_T-m^{-1})e^{\lambda_B\zeta_u}g_{\widehat h_T}(\zeta_u) \\
 & = Th_T(\beta)^{\varpi_B/2}\widehat m_T^{-1}r_T \zeta_ue^{\vartheta_T\zeta_u}g_{\widehat h_T}(\zeta_u)+w_T(B)^{-1}e^{-\lambda_BT/2}e^{\lambda_B\zeta_u}r'_Tg_{\widehat h_T}(\zeta_u) \\
& = I + II,
 \end{align*}
say, with $\vartheta_T \in [\min\{\lambda_B,\widehat \lambda_T\}, \max\{\lambda_B,\widehat \lambda_T\}]$.
%
%
Since $g \in \mathcal C^1 \subset \mathcal C$ and $\widehat m^{-1}_T$, $\vartheta_T$ and $\widehat h_T$ are $\mathcal B$-tight, we can write
$$|I| \leq T h_T(\beta)^{\varpi_B/2} \widehat m^{-1}_T r_T (C\widehat h_T +x) e^{\vartheta_T (C \widehat h_T + x)} |g_{\widehat h_T}(\zeta_u)|  = T h_T(\beta)^{\varpi_B/2} |g_{\widehat h_T}(\zeta_u)|  \widetilde r_T$$
and 
$$|II| \leq h_T(\beta)^{\varpi_B/2}e^{\lambda_B (C\widehat h_T + x)}r'_T|g_{\widehat h_T}(\zeta_u)|  
= h_T(\beta)^{\varpi_B/2} |g_{\widehat h_T}(\zeta_u)| \widetilde r'_T$$
where $\widetilde r_T$ and $\widetilde r'_T$ are $B \in \mathcal B$-tight. \\

\noindent{\it Step 2}.  We are left to proving the tightness of $T h_T(\beta)^{\varpi_B/2}|g_{\widehat h_T}(\zeta_u)|$ when averaging over $\mathring{\mathcal T}_T$ that is to say the tightness of $T h_T(\beta)^{\varpi_B/2}\mathcal E^T(\mathring{\mathcal T}_T, |g_{\widehat h_T}|)$. We plan to use Proposition~\ref{lem: kolmo}.  For $\kappa >0$, on the event
$${\mathcal A}_{T,\kappa} = \big\{\widehat h_T \in \mathcal I_{T,\kappa}\big\},\;\;\mathcal I_{T, \kappa} = \big[h_T(\beta)(1\pm \kappa T^2 v_T(B))\big],$$
we have
$$T h_T(\beta)^{\varpi_B/2}\mathcal E^T(\mathring{\mathcal T}_T, |g_{\widehat h_T}|) \leq III + IV,$$
with
$$III= Th_T(\beta)^{\varpi_B/2}\sup_{h \in \mathcal I_{T,\kappa}}\mathring{\mathcal E}_B(|g_{h}|)$$
and
\begin{align*}
IV & = Th_T(\beta)^{\varpi_B/2} \big(h_T(\beta)(1-\kappa T^2 v_T(B))\big)^{-\varpi_B/2} \sup_{h \in \mathcal I_{T,\kappa}}\big|\mathcal E^T\big(\mathring{\mathcal T}_T,h^{\varpi_B/2}|g_{h}|) - \mathring{\mathcal E}_B(h^{\varpi_B/2}|g_{h}|)\big| \\
& \leq T\sup_{h \in \mathcal I_{T,\kappa}}\big|\mathcal E^T\big(\mathring{\mathcal T}_T,h^{\varpi_B/2}|g_{h}|) - \mathring{\mathcal E}_B\big(h^{\varpi_B/2}|g_{h}|\big)\big|.
\end{align*}
Concerning the main term $III$, we write
\begin{align*}
\mathring{\mathcal E}_B(|g_{h}|) & = m\int_0^\infty h^{-1}|g\big(h^{-1}(x-y)\big)|e^{-\lambda_By}f_B(y)dy 
& \leq m\sup_y\big(e^{-\lambda_By}f_B(y)\big) \int_0^\infty |g(y)|dy  \lesssim 1
\end{align*}
since $B \in \mathcal B$, so we have a bound that does not depend on $h$ and we readily conclude $III \lesssim 1$ on $\mathcal A_{T,\kappa}$. For the remainder term $IV$, we apply Proposition~\ref{lem: kolmo} and obtain the ${\mathcal B}$-tightness of $IV$ (that actually goes to $0$ at a fast rate) on $\mathcal A_{T, \kappa}$.\\

\noindent {\it Step 3}. It remains to control the probability of $\mathcal A_{T,\kappa}$. By Proposition~\ref{prop: conv lambda}, we have
$\widehat \lambda_T = \lambda_B + Tv_T(B) r_T$, where $r_T$ is $\mathcal B$-tight.  It follows that
\begin{align*}
\PP(\mathcal A_{T,\kappa}^c) & = \PP\big(|\widehat h_T-h_T(\beta)| \geq  \kappa h_T(\beta)T^2v_T(B) \big) \\
& = \PP\big(\big|1-e^{-(\widehat \lambda_T-\lambda_B)\frac{1}{2\beta+1}T}\big| \geq  \kappa T^2v_T(B)\big) \\
& = \PP\big(|\tfrac{1}{2\beta+1}r_Te^{-\vartheta_T \frac{1}{2\beta+1} T}| \geq \kappa\big)  
\end{align*}
where both $|\vartheta_T| \leq |\widehat \lambda_T-\lambda_B|$ and $r_T$ are tight, and this term can be made arbitrarily small by taking $\kappa$ large enough.\\

The case $G_T(y) = {\bf 1}_{\{y \leq x\}}$ is obtained in the same way and is actually much simpler, since there is no factor $\widehat h_T^{-1}$ in the Step 2 which is therefore straightforward and there is also no need for a Step 3. We omit the details.
\end{proof}
\begin{proof}[Proof of Theorem~\ref{thm: upper rate}] We are ready to prove the main result of the paper. The key ingredient is Proposition~\ref{lem: kolmo}.\\

\noindent {\it Step 1.}  In view of Lemma~\ref{lem: degager est preli} with test function $g=K$, it is now sufficient to prove Theorem~\ref{thm: upper rate} replacing $\widehat B_T(x)$ by $\widetilde B_T(x)$, where
$$\widetilde B_T(x) = \frac{\mathcal E^T\big(\mathring{\mathcal T}_T, m^{-1}e^{\lambda_B\cdot}K_{\widehat h_T}(x-\cdot)\big)}{1-\mathcal E^T(\mathring{\mathcal T}_T, m^{-1}e^{\lambda_B\cdot}{\bf 1}_{\{\cdot \leq x\}})}.$$
Since $(x,y) \leadsto x/(1-y)$ is Lipschitz continuous on compact sets that are bounded away from $\{y=1\}$, this 
simply amounts to show the $\mathcal B$-tightness of 
\begin{equation} \label{den final}
w_T(B)^{-1}\Big(\mathcal E^T(\mathring{\mathcal T}_T, m^{-1}e^{\lambda_B\cdot}{\bf 1}_{\{\cdot \leq x\}})- \mathring{\mathcal E}_B(m^{-1}e^{\lambda_B\cdot}{\bf 1}_{\{\cdot \leq x\}}\big)\Big)
\end{equation}
and
\begin{equation} \label{numfinal}
w_T(B)^{-1}\big(\mathcal E^T\big(\mathring{\mathcal T}_T, m^{-1}e^{\lambda_B\cdot}K_{\widehat h_T}(x-\cdot)\big)-f_B(x)\big),
\end{equation}
where $w_T(B)^{-1}=(T^{-1/2})^{{\bf 1}_{\{\lambda_B=2\rho_B\}}}e^{\min\{\rho_B,\lambda_B/2\}T}h_T(\beta)^{\varpi_B/2}=v_T(B)^{-1}h_T(\beta)^{\varpi_B/2}$. We readily obtain the $\mathcal B$-tightness of \eqref{den final} by applying Theorem~\ref{rate avant T} with test function $g(y)=m^{-1}e^{\lambda_B y}{\bf 1}_{\{y \leq x\}}$ since $v_T(B) \ll w_T(B)$ (we even have convergence to $0$).\\

\noindent {\it Step 2}. We turn to the main term \eqref{numfinal}. 
For $h>0$, introduce the notation
$$K_hf_B(x)= \mathring{\mathcal E}_B(m^{-1}e^{\lambda_B\cdot}K_{h}) = \int_0^\infty K_{h}(x-y)f_B(y)dy.$$
For $\kappa>0$ on the event
${\mathcal A}_{T,\kappa} = \big\{\widehat h_T \in \mathcal I_{T,\kappa}\big\}$ with $\mathcal I_{T,\kappa} = \big[h_T(\beta)(1\pm \kappa T^2 v_T(B))\big]$
introducing the 	approximation term $K_hf_B(x)$, we obtain a bias-variance bound that reads
\begin{align*}
\big|\mathcal E^T(\mathring{\mathcal T}_T, m^{-1}e^{\lambda_B\cdot}K_{\widehat h_T})-f_B(x)\big| \leq I+II,
\end{align*}
with
$$I = \sup_{h \in \mathcal I_{T,\kappa}}\big|K_hf_B(x)-f_B(x)\big|$$
and
$$II = \sup_{h \in \mathcal I_{T,\kappa}}\big|\mathcal E^T(\mathring{\mathcal T}_T, m^{-1}e^{\lambda_B\cdot}K_{h})- \mathring{\mathcal E}_B(m^{-1}e^{\lambda_B\cdot}K_{h})\big|.$$
The term $I$ is treated by the following classical argument in nonparametric estimation: since $B \in \mathcal H_\mathcal D^\beta(L)$ we also have $f_B \in \mathcal H_\mathcal D^\beta(L')$ for another constant $L'$ that only depends on $\mathcal D, L$ and $\beta$. Write $\beta=\lfloor \beta \rfloor+\{\beta\}$ with $\lfloor \beta \rfloor$ a non-negative integer, $\{\beta\}>0$. By a Taylor expansion up to order $\lfloor \beta \rfloor $ (recall that the number $n_0$ of vanishing moments of $K$ in Assumption~\ref{kernel} satisfies $n_0>\beta$), we obtain
$$
I \lesssim \sup_{h \in \mathcal I_{T,\kappa}}h^\beta = \big( h_T(\beta)(1 + \kappa T^2 v_T(B)\big)^{\beta} \lesssim w_T(B)
$$
see for instance, Proposition 1.2 in Tsybakov~\cite{Tsybakov}. This term has the right order whenever $\lambda_B \leq \rho_B$ and is negligible otherwise.\\

\noindent {\it Step 3}. We further bound the term $II$ on ${\mathcal A}_{T, \kappa}$ as follows:
\begin{equation*} 
|II| \leq \big(h_T(\beta)(1-\kappa T^2v_T(B))\big)^{-\varpi_B/2} \sup_{h \in \mathcal I_{T,\kappa}}\big|\mathcal E^T(\mathring{\mathcal T}_T, h^{\varpi_B/2}m^{-1}e^{\lambda_B\cdot}K_{h})- \mathring{\mathcal E}_B(h^{\varpi_B/2}m^{-1}e^{\lambda_B\cdot}K_{h})\big|.
\end{equation*}
By assumption, we have $\beta \geq 1/2$, so by Proposition~\ref{lem: kolmo} applied to $f(y)=m^{-1}e^{\lambda_By}{\bf 1}_{\{y \leq x+C\}} \in \mathcal L_{C+x}$ and $g=K \in \mathcal C_{C+x}$ we conclude that $v_T(B)^{-1}h_T(\beta)^{\varpi_B/2}|II|$ is $\mathcal B$-tight. The fact that $v_T(B)^{-1}h_T(\beta)^{\varpi_B/2}=w_T(B)^{-1}$ enables us to conclude.\\

\noindent {\it Step 4}. It remains to control the probability of ${\mathcal A}_{T, \kappa}$. This is done exactly in the same way as for Step 3 in the proof of Lemma~\ref{lem: degager est preli}.
\end{proof}

\subsection{Proof of Theorem~\ref{thm: borne inf}}

We will prove actually a slightly stronger result, by restricting the supremum in $B$ over a neighbourhood of an arbitrary function $B_0$, provided $B_0$ is an element of the set $\mathbb B_{b,m}$ defined in \eqref{def top smooth} and slightly smoother in $\mathcal H_\mathcal D^\beta$ norm (and not identically equal to the maximal element of $\mathbb B_{b,m}$). (Remember also that $\mathbb B_{b,m} \subset \mathcal B^+_{b,m/(m-1)}$ by Proposition \ref{classe B plus}.) 

\vip

Remember that the evolution of the Bellman-Harris model can be described by a piecewise deterministic Markov process 
$$
X(t)=\big(X_1(t), X_2(t),\ldots\big),\quad t\geq 0
$$
with values in $\mathcal S=\bigcup_{k \geq 1}[0,\infty)^k $ and where the $X_i(t)$ denote the (ordered) ages of the living particles at time $t$.
Following L\"ocherbach~\cite{eva1}, we set $\mathbb{D}([0,\infty),\mathcal S)$ for the Skorokhod space of c\`adl\`ag functions $\varphi: [0,\infty) \rightarrow \mathcal S$ and introduce the subset $\Omega \subset \mathbb{D}([0,\infty),\mathcal S)$ of functions $\varphi$ such that:
\begin{enumerate}
\item[(i)] There is an increasing sequence of jump times $T_0 = 0 < T_1 < T_2 < \cdots$ such that the restriction $\varphi_{\big|[T_k,T_{k+1})}$ is continuous with values in $[0,\infty)^{l_{k,\varphi}}$ for some $l_{k,\varphi}\geq 0$ and every $k \geq 0$.
\item[(ii)]  We have $\ell\big( \varphi(T_k) \big) \neq \ell\big( \varphi(T_{k+1}) \big)$ for every $k\geq 0$, where we set $\ell(x)= \sum_{k \geq 0} k {\bf 1}_{\{x \in [0,\infty)^k\}}$ for $x \in \mathcal S$.
\end{enumerate}
We endow $\Omega$ with its Borel sigma-field $\mathcal F$, its canonical process $X_t(\varphi) = (\varphi_1(t), \varphi_2(t),\ldots)$ and its canonical filtration $(\mathcal{F}_t)_{t\geq 0}$ (modified in order to be right-continuous).
By Proposition 3.3 of L\"ocherbach~\cite{eva1}, there is a unique probability measure $\mathbb P_B$ on $(\Omega, \mathcal F, (\mathcal F_t)_{t \geq 0})$ such that $X$ is strongly Markov under $\mathbb P_B$ with $\PP_B(X(0)=0)=1$ ({\it i.e.} we start with one common ancestor with age $0$ at time $0$) and such that the random continuous time rooted tree associated to $X$ via 
$$\sum_{i \geq 1}{\bf 1}_{\{X_i(t)>0\}}\delta_{X_i(t)} = \sum_{u \in \mathcal T}{\bf 1}_{\{t \in [b_u,d_u)\}}\delta_{t-b_u}$$
is a Harris-Bellman process according to Definition~\ref{BellmanHarris}. The strategy for proving the lower bound is a classical two point information inequality: we nevertheless need to be careful since the target lower bound rate $e^{-\lambda_B\frac{\beta}{2\beta+1}T}$ is parameter dependent in a non-trivial way.

\vip

\noindent {\it Step 1.} Let $\delta >0$. Fix $B_0 \in \mathbb B_{b,m} \cap \mathcal H_\mathcal D^\beta(L-\delta)$ and $x \in \mathcal D$. Then, for large enough $T$, setting $h_T(B)=e^{-\lambda_B\frac{1}{2\beta+1}T}$, we construct a perturbation $B_T$ of $B_0$ defined by
$$B_T(y) = B_0(y) + ah_T(B_0)^{\beta+1} K_{h_T(B_0)}\big(y-x\big), \quad y \in [0,\infty),$$
for some nonnegative smooth kernel $K$ with compact support such that $K(0)=1$ and for some $a=a_{\delta, K} >0$ chosen in such a way that $B_T \in \mathbb B_{b,m} \cap \mathcal H^\beta_{\mathcal D}(L)$ for every $T \geq 0$. Such a choice is always possible (if $B_0 \neq \max\{C,1\}$ identically in a neighbourhood of $x$, which we may and will assume from now on) thanks to the assumption $\|B_0\|_{{\mathcal H}^\beta_{\mathcal D}}\leq L-\delta$; it suffices then to impose $\|ah_T^{\beta+1} K_{h_T}(\cdot-x)\|_{{\mathcal H}^\beta_{\mathcal D}} \leq \delta$ which is easily obtained by picking $a_{\delta, K}$ sufficiently small. 

Also, by construction, we have $B_0(y) \leq B_T(y)$ for every $y\geq0$ hence $\lambda_{B_0} \leq \lambda_{B_T}$, compare the proof of Proposition~\ref{prop: convergence semigroupe} (ii) and at $y=x$, the lower estimate $|B_0(x)-B_T(x)| = a_{\delta, K}h_T^{\beta}(B_0)$ holds, and this quantity is of order $e^{-\lambda_{B_0}\frac{\beta}{2\beta+1}T}$.

\vip

\noindent {\it Step 2.} Abusing notation slightly, we further write $\PP_B$ for ${\PP_B}_{| \mathcal F_T}$, {\it i.e.} the measure in restriction to the $\sigma$-field generated by the observation $(X(t))_{0 \leq t \leq T}$. Since $B_0,B_T \in \mathbb B_{b,m} \cap \mathcal H^\beta_{\mathcal D}(L)$, for an arbitrary estimator $\widehat B_T(x)$ and any constant $C'>0$ the maximal risk is bounded below by
\begin{align*}
& \max_{B \in \{B_0, B_T\}}\PP_B\big(e^{\lambda_B\frac{\beta}{2\beta+1}T}|\widehat B_T(x)-B(x)| \geq C'\big) \\
\geq  &\tfrac{1}{2}\left(\PP_{B_0}\big(e^{\lambda_{B_0}\frac{\beta}{2\beta+1}T}|\widehat B_T(x)-B_0(x)|\geq C'\big)+\PP_{B_T}\big(e^{\lambda_{B_T}\frac{\beta}{2\beta+1}T}|\widehat B_T(x)-B_T(x)|\geq C'\big)\right)\\
\geq &\, \tfrac{1}{2}\,\E_{B_0}\left[{\bf 1}_{\big\{e^{\lambda_{B_0}\frac{\beta}{2\beta+1}T}|\widehat B_T(x)-B_0(x)|\geq C'\big\}}+{\bf 1}_{\big\{e^{\lambda_{B_T}\frac{\beta}{2\beta+1}T}|\widehat B_T(x)-B_T(x)|\geq C'\big\}}\right]-\tfrac{1}{2}\|\PP_{B_0}-\PP_{B_T}\|_{TV}.
\end{align*}
By triangle inequality, we have
\begin{align*}
& e^{\lambda_{B_0}\frac{\beta}{2\beta+1}T}|\widehat B_T(x)-B_0(x)|+e^{\lambda_{B_T}\frac{\beta}{2\beta+1}T}|\widehat B_T(x)-B_T(x)| \\
\geq & e^{\min\{\lambda_{B_0}, \lambda_{B_T}\}\tfrac{\beta}{2\beta+1}T}|B_0(x)-B_T(x)| \geq a_{K,\delta}
\end{align*} 
by Step 1, so if we pick $C' < a_{K, \delta}/2$, one of the two indicators within the expectation above must be equal to one with full $\PP_{B_0}$-probability. In that case
$$\max_{B \in \{B_0, B_T\}}\PP_B\big(e^{\lambda_B\frac{\beta}{2\beta+1}T}|\widehat B_T(x)-B(x)| \geq C'\big) \geq \tfrac{1}{2}(1-\|\PP_{B_0}-\PP_{B_T}\|_{TV})$$
and Theorem~\ref{thm: borne inf} is thus proved if 
$\limsup_{T \rightarrow \infty}\|\PP_{B_0}-\PP_{B_T}\|_{TV}<1$.

\noindent {\it Step 3.} By Pinsker's inequality, we have
$\|\PP_{B_0}-\PP_{B_T}\|_{TV} \leq \frac{\sqrt{2}}{2}\,\Big(\E_{B_0}\Big[\log\frac{d\PP_{B_0}}{d\PP_{B_T}}\Big]\Big)^{1/2}
$. By Theorem~3.5 in~\cite{eva1}, the measures $\PP_{B_0}$ and $\PP_{B_T}$ are equivalent on $\mathcal F_T$ and we have
$$
\log\Big(\frac{d\mathbb{P}_{B_T}}{d\mathbb{P}_{B_0}}\Big) = \sum_{u \in  \mathring{\mathcal{T}}_T} \log\bigg(\frac{B_T}{B_0}(\zeta_u)\bigg) - \int_0^T \sum_{u \in \partial \mathcal{T}_s} (B_T-B_0)(\zeta_u^s) \, ds,
$$
where $\zeta_u^t$ denotes the age of the cell $u$ at time $t \in I_u=[b_u,d_u)$. Using $-\log(1+x) \leq x^2 - x$ if $x \geq -1/2$ and setting $\varepsilon_T(y) = a_{K,\delta}h_T(B_0)^{\beta+1} K_{h_T(B_0)}\big(y-x\big)$, we further infer
\begin{align*}
\| \PP_{B_0} - \PP_{B_T} \|_{TV}^2 
& \leq  \frac{1}{2} \big(  \E_{B_0}\big[  \sum_{u \in  \mathring{\mathcal{T}}_T} \frac{\varepsilon_T^2}{B_0^2}(\zeta_u) \big] - \E_{B_0}\big[  \sum_{u \in  \mathring{\mathcal{T}}_T} \frac{\varepsilon_T}{B_0}(\zeta_u) \big]  +  \int_0^T  \E_{B_0}\big[  \sum_{u \in \partial \mathcal{T}_s} \varepsilon_T(\zeta_u^s) \big] ds \big)  \\
& = \frac{1}{2m} \int_0^T e^{\lambda_{B_0} s}\, \E_{B_0}\Big[ \frac{\varepsilon_T^2}{B_0^2}(\widetilde \chi_s) H_{B_0}(\widetilde \chi_s) \Big] ds
\end{align*}
by  \eqref{Mt1 bord} and \eqref{Mt1 int} in Proposition~\ref{prop: Mt1} and the fact that the last two terms cancel. We now use the same kind of estimates as in the proof of Proposition~\ref{lem: vanishing test}, Step 1 with test function $g = \varepsilon_T/B_0$ to finally get
\begin{align*}
\| \PP_{B_0} - \PP_{B_T} \|_{TV}^2 & \lesssim e^{\lambda_{B_0} T} \big| B_0^{-1}\varepsilon_T\big|_2^2 + \big| B_0^{-1}\varepsilon_T \big|_{\infty}^2
\lesssim a_{K,\delta}^2
\end{align*}
and this term can be made arbitrarily small by picking $a_{K,\delta}$ small enough.

\subsection{Proof of Proposition \ref{classe B plus}} Pick  $B \in \mathbb B_{b,m}$. We need to prove that $\lambda_B \leq \rho_B = \inf_x H_B(x)$. 
By representation \eqref{hazard formula}, we have 
\begin{align*}
H_B(x) & =\frac{me^{-\lambda_B x}f_B(x)}{1-m\int_0^x e^{-\lambda_By}f_B(y)dy} \\
& = \frac{me^{-\lambda_B x}B(x)e^{-\int_0^x B(y)dy}}{1-m\int_0^x e^{-\lambda_By}B(y)e^{-\int_0^y B(u)du}dy}.
\end{align*}
Set 
$$G_B(x) = me^{-\lambda_B x}B(x)e^{-\int_0^x B(y)dy}- \lambda_B\big(1-m\int_0^x e^{-\lambda_By}B(y)e^{-\int_0^y B(u)du}dy\big).$$
The statement  $\lambda_B \leq \rho_B$  is equivalent to proving that $\inf_{x \geq 0}G_B(x) \geq 0$.  We first claim that 
$$B(x) \leq \widetilde B(x)\;\;\text{for every}\;\;x\in (0,\infty)\;\;\text{implies}\;\;\lambda_B \leq \lambda_{\widetilde B}.$$
Indeed, in that case, one can construct on the same probability space two random variables $\tau_B$ with density $f_B$ and $\tau_{\widetilde B}$ with density $f_{\widetilde B}$ such that $\tau_B \geq \tau_{\widetilde B}$. It follows that $\phi_B(\lambda)=\E[e^{-\lambda \tau_B}] \leq \phi_{\widetilde B}(\lambda) = \E[e^{-\lambda \tau_{\widetilde B}}]$ for every $\lambda \geq 0$. Also, $\phi_B$ and $\phi_{\widetilde B}$ are both non-increasing, vanishing at infinity, and $\phi_B(0)=\phi_{\widetilde B}(0)=1 > \tfrac{1}{m}$. Consequently, the values $\lambda_B$ and $\lambda_{\widetilde B}$
such that $\phi_B(\lambda_B)=\phi_{\widetilde B}(\lambda_{\widetilde B}) = \tfrac{1}{m}$ necessarily satisfy $\lambda_B \leq \lambda_{\widetilde B}$ hence the claim.
Now, for constant functions $B(x)= \alpha$, we clearly have $\lambda_B=(m-1)\alpha$ and this enables us to infer
$$\lambda_B \leq (m-1) \sup_xB(x).$$
Remember now that $B \in \mathbb B_{b,m}$ implies $b \leq B(x) \leq \tfrac{m}{m-1}b$ for every $x \geq 0$.  Therefore
\begin{equation} \label{controle lambda}
\lambda_B \leq (m-1) \frac{m}{m-1} b = m b \leq m B(0)
\end{equation}
and  $G_B(0)=mB(0) - \lambda_B \geq 0$ follows. Moreover, one readily checks that
$$G'_B(x) = m e^{-\lambda_B x} e^{-\int_0^xB(y)dy}\big(B'(x)-B(x)^2\big) \leq 0$$
since $B'(x)-B(x)^2\leq 0$ as soon as $B \in \mathbb B_{b,m}$. So $G_B$ is non-increasing, $G_B(0)\geq 0$ and its infimum is thus attained for $x \rightarrow \infty$. Since $G_B(\infty)=0$, we conclude $\inf_{x \geq 0}G_B(x) \geq 0$.

We finally briefly indicate how to show that $\mathcal B_{b,C}^-$ is non-trivial when $C>mb/(m-1)$. To that end, pick $0<x_0 \leq x_1$, $mb/(m-1) < c \leq C$ and let 
$B(x)=b$ for $x \leq x_0$, $B(x)=c$ for $x \geq x_1$ and any smooth continuation between $x_0$ and $x_1$ bounded above by $C$ and below by $b$. Then, having $b,c$ such that $2m(m+2)b/(m-1)<c$ and suitable choices for $x_0$ and $x_1$ implies $\rho_B < \lambda_B/2$. Having $2mb/(m-1)>c$ and suitable choices for $x_0,x_1$ implies $\rho_B < \lambda_B \leq 2\rho_B$. The computations, based on the same kind of estimates, are rather tedious but not difficult. We omit the details.

\section{\textsc{Appendix}} \label{sec: appendix}

\subsection{Heuristics for the convergences to the limits \eqref{def limite int} and \eqref{def limite bord}}

\subsubsection*{Information from $\mathcal E^T(\partial {\mathcal T}_T, g)$}  \label{heuristic}
 Heuristically, we postulate for large $T$ the approximation
$$\mathcal E^T(\partial {\mathcal T}_T,g) \sim 
\frac{1}{\E[|\partial {\mathcal T}_T|]}\E\Big[\sum_{u \in \partial {\mathcal T}_T}g(\zeta_u^{T})\Big].$$
Then, a classical result  based on renewal theory (see Theorem 17.1 pp 142-143 of \cite{Harris}) gives the estimate
\begin{equation} \label{approx temp}
\E\big[|\partial{\mathcal T}_T|\big] \sim \kappa_{B}e^{\lambda_{B}T},
\end{equation}
where $\lambda_{B}>0$ is the Malthus parameter defined in \eqref{def malthus} and $\kappa_{B}>0$ is an explicitly computable constant (that also depends on $m$, see~\cite{Harris} and also Lemma~\ref{estimate moment un} below). As for the numerator,
call $\chi_t$ the age of a particle at time $t$ along a branch of the tree picked at random uniformly at each branching event. The process $(\chi_t)_{t\geq 0}$ is  Markov process with values in $[0,\infty)$ with infinitesimal generator 
\begin{equation} \label{def generator}
\mathcal A_B g(x) = g'(x)+B(x)\big(g(0)-g(x)\big)
\end{equation}
densely defined on continuous functions vanishing at infinity. Assume for simplicity that each cell $u \in \mathcal U$ has exactly $m$ children at each division. 
It is then relatively straightforward to obtain the identity
\begin{equation} \label{firstMtO}
\E\Big[\sum_{u \in \partial {\mathcal T}_T}g(\zeta_u^{T})\Big] = \E\big[m^{N_T}g(\chi_T)\big],
\end{equation}
where $N_t = \sum_{s \leq t}{\bf 1}_{\{\chi_s-\chi_{s_-} < 0\}}$ is the counting process associated to $(\chi_t)_{t \geq 0}$, see Proposition~\ref{prop: Mt1} in a general setting. Putting together \eqref{approx temp} and \eqref{firstMtO}, we thus expect
$$\mathcal E^T(\partial {\mathcal T}_T,g) \sim \kappa_{B}^{-1}e^{-\lambda_{B}T}\E\big[m^{N_T}g(\chi_T)\big],$$
and we anticipate that the term $e^{-\lambda_B T}$ should somehow be compensated by the term $m^{N_T}$ within the expectation. To that end, following Cloez~\cite{cloez} (and also in Bansaye {\it et al.}~\cite{BDMT} when $B$ is constant) one introduces an auxiliary ``biased" Markov process $(\widetilde \chi_t)_{t\geq 0}$, with generator $\mathcal A_{H_B}$ for a biasing function $H_B(x)$ 
characterised by
\begin{equation} \label{characterisation H_B}
f_{H_B}(x) 
= me^{-\lambda_Bx}f_B(x),\;x \geq 0,
\end{equation}
where 
$f_B(x)=B(x)\exp(-\int_0^x B(y)dy)$ denotes the density associated to the division rate $B$, as follows from \eqref{hazard formula} or \eqref{def B}. This implies
$$
H_B(x)=\frac{me^{-\lambda_Bx}f_B(x)}{1-m\int_0^xe^{-\lambda_By}f_B(y)ds}.
$$
Furthemore, this choice (and this choice only, see Proposition~\ref{prop: Mt1}) enables us to obtain 
\begin{equation} \label{first HB}
e^{-\lambda_{B}T}\E\big[m^{N_T}g(\chi_T)\big] = m^{-1}\E\Big[g(\widetilde \chi_T)B(\widetilde \chi_T)^{-1} H_B(\widetilde \chi_T)\Big]
\end{equation}
with $\widetilde \chi_0=0$ under $\PP$. Moreover $(\widetilde \chi_t)_{t \geq 0}$ is geometrically ergodic, with invariant probability $c_B\exp(-\int_0^xH_B(y)dy)dx$ (see Proposition~\ref{prop: convergence semigroupe}). We further anticipate  
\begin{align*}
\E\Big[g(\widetilde \chi_T) B(\widetilde \chi_T)^{-1} H_B(\widetilde \chi_T)\Big] & \sim c_B\int_0^\infty g(x)B(x)^{-1}H_B(x)e^{-\int_0^xH_B(y)dy}dx \\
& = mc_B\int_0^\infty g(x)e^{-\lambda_B x}B(x)^{-1}f_B(x)dx 
\end{align*}
assuming everthing is well-defined, since $H_B(x)\exp(-\int_0^xH_B(y)dy)=f_{H_B}(x)=me^{-\lambda_B}f_B(x)$ by \eqref{characterisation H_B}.
Finally, we have $\kappa_B^{-1}c_B = \lambda_B\frac{m}{m-1}$ by Lemma~\ref{estimate moment un} which enables us to conclude
$$
\mathcal E^T(\partial {\mathcal T}_T,g) \sim \partial \mathcal E_B(g),
$$ 
where
$$
 \partial \mathcal E_B(g) = \lambda_B\frac{m}{m-1} \int_0^\infty g(x)e^{-\lambda_Bx} e^{-\int_0^x B(y)dy} dx. 
$$
Unfortunately, the statistical information extracted from $\mathcal E^T(\partial {\mathcal T}_T, g)$ does not enable us to obtain classical optimal rates of convergence, since the form of $\partial \mathcal E_B(g)$ involves an antiderivative of $B$ leading to so-called ill-posedness. This is discussed at length in Section~\ref{discussion}. We thus investigate in a second step the statistical information we can get from $\mathring {\mathcal T}_T$.

\subsubsection*{Information from $\mathcal E^T(\mathring {\mathcal T}_T, g)$} The situation is a bit different if we allow for data in  $\mathring {\mathcal T}_T$. 
Note first that $\zeta_u^T=\zeta_u$ on the event $u \in \mathring{\mathcal{T}}_T$.
We also have in that case a many-to-one formula that now reads 
\begin{equation} \label{many-to-one cloez}
\E\Big[\sum_{u \in \mathring{\mathcal T}_T}g(\zeta_u^T)\Big]  = \E\Big[\sum_{u \in \mathring{\mathcal T}_T}g(\zeta_u)\Big] = m^{-1}\int_0^T e^{\lambda_B s}\E\big[g(\widetilde \chi_s)H_B(\widetilde \chi_s)\big]ds,
\end{equation}
where $(\widetilde \chi_t)_{t \geq 0}$ is the one-dimensional auxiliary Markov process 
with generator $\mathcal{A}_{H_B}$, see \eqref{def generator}, where $H_B$ is characterised by \eqref{characterisation H_B} above. Assuming again ergodicity, we approximate the right-hand side of \eqref{many-to-one cloez} and obtain
\begin{align*} 
\E\Big[\sum_{u \in \mathring{\mathcal T}_T}g(\zeta_u)\Big]  & \sim c_Bm^{-1} \frac{e^{\lambda_BT}}{\lambda_B} \int_0^\infty g(x)H_B(x)e^{-\int_0^x H_B(u)du}dx \\
 & = c_B \frac{e^{\lambda_BT}}{\lambda_B} \int_0^\infty g(x)e^{-\lambda_B x}f_{B}(x)dx \nonumber 
\end{align*}
since $H_B(x)\exp(-\int_0^x H_B(y)dy)=f_{H_B}(x)=me^{-\lambda_Bx}f_B(x)$ by \eqref{characterisation H_B}.
We again have an approximation of the type  \eqref{approx temp} with another constant $\kappa_B'$, see Lemma~\ref{estimate moment un avant T}  and we eventually expect
$$
\mathcal E^T(\mathring{\mathcal T}_T, g)\sim \mathring{\mathcal E}_B\big(g\big), 
$$
where
$$\mathring{\mathcal E}_B\big(g\big)  = \frac{c_B}{\lambda_B\kappa_B'}\int_0^\infty g(x)e^{-\lambda_B x}f_{B}(x)dx =  m\int_0^\infty g(x)e^{-\lambda_B x}f_{B}(x)dx$$
as $T \rightarrow \infty$, where the last equality stems from the identity $c_B=\lambda_B\kappa'_Bm$ that can be readily derived by picking $g=1$ and using \eqref{characterisation H_B} together with the fact that $f_{H_B}$ is a density function.

\subsection{Proof of Proposition~\ref{prop: Mt1}}
We start with a continuous time rooted tree which is a Bellman Harris process in the sense of Definition~\ref{BellmanHarris}, so we have random variables $(\zeta_u, \nu_u, u \in \mathcal U)$ satisfying properties (i), (ii) and (iii) of the definition. For $u \in \mathcal U$, and $t\geq 0$, let 
$
\Lambda_t^u = \sum_{v \prec u(t)} \log(\nu_v), t \geq 0
$
denote the process that encodes the birth times and the numbers of children of the ancestors of $u$. 
Let $\vartheta = (\vartheta_k)_{k \geq 0}$ with  $\vartheta_k \in \mathcal U$  be such that 
$|\vartheta_k|=k$ for $k \geq 1$ (with $\vartheta_0=\varnothing$) and $\vartheta_k \preceq \vartheta_l$ for $k \leq l$. We associate to $\vartheta$ a counting process $(N_t)_{t \geq0}$ via the relationship
$$
b_{\vartheta_{N_t}} \leq t < d_{\vartheta_{N_t}}, \quad t \geq 0.
$$
This enables us to further obtain a ``tagged process of age" such that
$\chi_t = \zeta^t_{\vartheta_{N_t}}$ for $t \in I_{\vartheta_{N_t}}$ and
also a process $(\Lambda_t)_{t \geq 0}$ that encodes the genealogy of the tagged branch
$$
\Lambda_t = \sum_{k = 1}^{N_t} \log(\nu_{\vartheta_k}), \quad t \geq 0.
$$
\noindent {\it Step 1.} Let us pick $\vartheta$ at random along the genealogical tree ${\mathcal T}$. This means that if ${\mathcal H}_n$ denotes the sigma-field generated by $(\zeta_u, \nu_u, u \in \mathcal T, |u|\leq n)$, then on the event $\{ t \in I_u \}$ (\textit{i.e.} the particle $u$ is living at time $t$), we have (or rather, we set)
$$
{\mathbb P}\big(\vartheta_{N_t} = u \big| {\mathcal H}_{|u|} \big) = \prod_{v \prec u} \frac{1}{\nu_{v}} = e^{-\Lambda^u_t}. 
$$
It is not difficult to see that $(\chi_t)_{t \geq 0}$ is a Markov process with generator $\mathcal A_B$. 
By definition of $(\chi_t)_{t \geq 0}$ and $(\Lambda_t)_{t \geq 0}$, it follows that ${\mathbb E}[e^{\Lambda_T}g(\chi_T)] $ can be rewritten as
\begin{align*}
\sum_{u \in {\mathcal U}} {\mathbb E} [e^{\Lambda_T}g(\chi_T) {\bf 1}_{\{T \in I_u, u = \vartheta_{N_T}\}}] & = \sum_{u \in {\mathcal U}} {\mathbb E} [e^{\Lambda_T^u}g(\zeta_u^T) {\bf 1}_{\{T \in I_u, u = \vartheta_{N_T}\}}]  = \sum_{u \in {\mathcal U}} {\mathbb E} [g(\zeta_u^T) {\bf 1}_{\{T \in I_u\}}],
\end{align*}
where the last equality is obtained by conditioning with respect to ${\mathcal H}_{|u|}$. 

\vip

\noindent {\it Step 2}. For $j \geq 1$, let 
$
\tau_j = \inf \{ t \geq 0, N_t \geq j \} - \inf \{ t \geq 0, N_t \geq j-1\}
$
denote the durations between the jumps of $(\chi_t)_{t \geq 0}$, so that
$$
e^{\Lambda_T} g(\chi_T) = \sum_{k = 0}^\infty e^{\sum_{j = 1}^k \log(\nu_{\vartheta_j})} g(T - \sum_{j = 1}^k \tau_j) {\bf 1}_{\big\{ \sum_{j = 1}^k \tau_j \leq T < \sum_{j = 1}^{k+1} \tau_j\big\}}.
$$
By properties (i)-(iii) of Definition~\ref{BellmanHarris}, the $\tau_i$ are independent with common distribution $f_B(x)dx$, and independent of the $\nu_{\vartheta_k}$ that are independent with common distribution $(p_k)_{k \geq 1}$.
We thus infer that $ \E[e^{\Lambda_T}g(\chi_T)]$ is equal to 
\begin{align*}
 \sum_{k = 0}^{\infty}\sum_{h_j \geq 1, j \leq k} e^{\sum_{j=1}^k \log(h_j)} \prod_{j = 1}^k p_{h_j}\int_{[0,\infty)^{k+1}} g(T - \sum_{j = 1}^k t_j) {\bf 1}_{\{ \sum_{j = 1}^k t_j \leq T < \sum_{j = 1}^{k+1} t_j\}} \prod_{j = 1}^{k+1} f_B(t_j) dt_1\ldots dt_{k+1}.
\end{align*}
We set ${\overline F}_B(x) = 1- \int_0^x f_B(y)dy$ and $q_k = m^{-1} k p_k$, so that $(q_k)_{k\geq 1}$ defines a probability distribution. Using $f_{H_B}(x) = m e^{-\lambda_B x} f_B(x)$, we can rewrite the preceding formula so that 
\begin{multline*}
e^{-\lambda_B T}  \E\big[e^{\Lambda_T}g(\chi_T)\big] = 
\sum_{k = 0}^{\infty} \sum_{h_j \geq 1, j \leq k} \prod_{j = 1}^k q_{h_j}  \int_{[0,\infty)^{k}} g(T - \sum_{j = 1}^k t_j) {\bf 1}_{\{ T - \sum_{j = 1}^k   t_j\geq 0 \}} e^{-\lambda_B (T - \sum_{j = 1}^k t_j)}  \\
\times {\overline F}_B(T-\sum_{j = 1}^{k} t_j)   \prod_{j = 1}^{k} f_{H_B}(t_j) dt_1\ldots dt_{k}.
\end{multline*}

\vip

\noindent {\it Step 3}.
Putting $W_B(x) = m e^{-\lambda_B x}  {\overline F}_B(x)/{\overline F}_{H_B}(x)$, we finally obtain the representation
$$
e^{-\lambda_B T}  \E\big[e^{\Lambda_T}g(\chi_T)\big] = \frac{1}{m} \E\Big[g(\widetilde \chi_T) W_B(\widetilde \chi_T)\Big],
$$
where $(\widetilde \chi_t)_{t \geq 0}$ is a Markov process with generator $\mathcal A_{H_B}$ that can be constructed in the same way as  $(\chi_t)_{t \geq 0}$, substituting $f_B$ by $f_{H_B}$.  Straightforward computations give $W_B(x) = \frac{H_B(x)}{B(x)}$. Putting together all the three steps,  we have proved
$$\sum_{u \in {\mathcal U}} {\mathbb E} \big[g(\zeta_u^T) {\bf 1}_{\{T \in I_u}\big]= \E\big[e^{\Lambda_T}g(\chi_T)\big]  = \frac{e^{\lambda_BT}}{m}\E\Big[g(\widetilde \chi_T)\frac{H_B(\widetilde \chi_T)}{B(\widetilde \chi_T)}\Big].
$$
Noticing that $\sum_{u \in {\mathcal U}} {\mathbb E} [g(\zeta_u^T) {\bf 1}_{\{T \in I_u}]$ is nothing but $\E\big[\sum_{u \in \partial \mathcal T_T}g(\zeta_u^T)\big]$ establishes \eqref{Mt1 bord}.

\vip
 
\noindent {\it Step 4}. By definition of the set $\mathring {\mathcal T}_T$,
\begin{align*}
\E\Big[\sum_{u \in \mathring {\mathcal T}_T}g(\zeta_u)\Big]  & =  \sum_{u \in {\mathcal U}} \E\Big[g(\zeta_u) {\bf 1}_{\{b_u+\zeta_u\leq T\}} {\bf 1}_{\{u \in {\mathcal T}\}} \Big] .
\end{align*}
We denote by $\mathcal F_t$ the sigma-field generated by $(\zeta_u^s, u \in \partial \mathcal T_s, s \leq t)$ and we note that $d_u{\bf 1}_{\{u \in \mathcal T\}}$ is a stopping time for the filtration $(\mathcal F_t)_{t \geq 0}$. Conditioning w.r.t $\mathcal F_{b_u}$, using that the $\zeta_u$ are independent of  $\mathcal F_{b_u}$, we successively obtain
\begin{align*}
\E\Big[\sum_{u \in \mathring {\mathcal T}_T}g(\zeta_u)\Big]  & =  \sum_{u \in {\mathcal U}} \E\Big[{\bf 1}_{\{u \in {\mathcal T}\}} \int_0^{\infty} g(x) {\bf 1}_{\{b_u+x\leq T\}} B(x) e^{-\int_0^x B(y) dy}  dx \Big] \\
& =  \sum_{u \in {\mathcal U}} \E\Big[ {\bf 1}_{\{u \in {\mathcal T}\}}\int_0^{\infty} \Big( \int_0^y g(x) B(x) {\bf 1}_{\{b_u+x\leq T\}} dx \Big) B(y) e^{-\int_0^y B(z) dz}  dy\Big] \\
& =  \sum_{u \in {\mathcal U}} \E\Big[ {\bf 1}_{\{u \in {\mathcal T}\}} \int_0^{\zeta_u} g(x) B(x) {\bf 1}_{\{b_u+x\leq T\}}dx\Big] \\
 & =  \sum_{u \in {\mathcal U}} \E\Big[{\bf 1}_{\{u \in {\mathcal T}\}}  \int_{b_u}^{d_u} g(\zeta_u^s ) B(\zeta_u^s ) {\bf 1}_{\{s\leq T\}} ds \Big].
\end{align*}
using that  $\zeta_u^s = s - b_u$ for $s \in I_u$ in order to obtain the last equality.
Finally, observing that $\{ s\in I_u\} = \{u \in \partial {\mathcal T}_s\}$, we finally infer
\begin{align*}
\E\Big[\sum_{u \in \mathring {\mathcal T}_T}g(\zeta_u)\Big]  & =  \int_0^\infty \E\Big[  \sum_{u \in \partial {\mathcal T}_s} g(\zeta_u^s ) B(\zeta_u^s )\Big] {\bf 1}_{\{s\leq T\}} ds. 
\end{align*}
Using \eqref{Mt1 bord} completes the proof of \eqref{Mt1 int}. 

\subsection{Proof of \eqref{bansaye forks} of Proposition~\ref{prop: Mt1 forks} }
Whenever $(u,v) \in \mathcal{FU}$ there exist $w,\tilde{u}$ and $\tilde{v} \in \mathcal{U}$ together with  integers $i\neq j$, such that $u = wi\tilde{u}$ and $v = wj\tilde{v}$. Conditioning w.r.t $\mathcal{F}_{d_w}$, using the branching property between descendants of $w$ and the strong Markov property at time $d_w$, we have
\begin{align*}
\E\Big[\sum_{(u,v)\in \mathcal{FT} \cap \mathring {\mathcal T}_T^2} g(\zeta_u)g(\zeta_v)\Big] & = \sum_{(u,v) \in  \mathcal{FU}} \E\Big[ g(\zeta_u) {\bf 1}_{\{d_u < T\}} {\bf 1}_{\{u \in \mathcal{T}\}} g(\zeta_v) {\bf 1}_{\{d_v < T\}} {\bf 1}_{\{v \in \mathcal{T}\}} \Big] \\
& = \sum_{w\in \mathcal{U}} \sum_{i \neq j} \E\Big[ \E\big[ \sum_{\tilde{u}\in \mathcal{U}} g(\zeta_{wi\tilde{u}}) {\bf 1}_{\{d_{wi\tilde{u}} < T\}} {\bf 1}_{\{wi\tilde{u} \in \mathcal{T}\}} \big| \mathcal{F}_{d_w} \big]  \\
& \qquad \qquad \qquad \times \E \big[ \sum_{\tilde{v}\in \mathcal{U}} g(\zeta_{wj\tilde{v}}) {\bf 1}_{\{d_{wj\tilde{v}} < T\}} {\bf 1}_{\{wj\tilde{v} \in \mathcal{T}\}} \big| \mathcal{F}_{d_w} \big] \Big] \\
& =  \sum_{w\in \mathcal{U}} \sum_{i \neq j} \E\Big[  {\bf 1}_{\{wi \in \mathcal{T}, wj \in \mathcal{T}\}} \big( \E\big[ \sum_{u \in \mathcal{T}} g(\zeta_u) {\bf 1}_{\{d_u < T-t\}} \big] _{|t = d_w}\big)^2 {\bf 1}_{\{ d_w < T \}} \Big].
\end{align*}
Notice that $\{ wi \in \mathcal{T}, wj \in \mathcal{T} \} = \{ w \in \mathcal{T}\} \cap \{ i \leq \nu_w, j \leq \nu_w\}$, and $\nu_w$ is independent of $d_w$ and has distribution $(p_k)_{k \geq 1}$. We conclude by using \eqref{Mt1 int} of Proposition~\ref{prop: Mt1} (slightly generalized for test functions that depend on $d_u$ and $\zeta_u$).
Let us now turn to \eqref{lineage}.
For $u,v \in \mathcal T$ with $u \prec v$, we have $uiw=v$ for some $w \in \mathcal T$ and some integer $i$. It follows that
\begin{align*}
\E\big[\sum_{u,v \in \mathring{\mathcal T}_T, \atop u \prec v}g(\zeta_u)g(\zeta_v)\big] & = \sum_{u \in \mathcal U} \sum_{i} \E\Big[g(\zeta_u){\bf 1}_{\{d_u < T\}} {\bf 1}_{\{u \in \mathcal{T}\}} \E\big[\sum_{w\in \mathcal U}g(\zeta_{uiw}){\bf 1}_{\{d_{uiw} < T\}}{\bf 1}_{\{uiw \in \mathcal{T}\}}\,\big| \mathcal F_{d_u}\big]\Big] \\
& = \sum_{u \in \mathcal U} \sum_{i} \E\Big[g(\zeta_u){\bf 1}_{\{0 \leq d_u < T\}} {\bf 1}_{\{ui \in \mathcal{T}\}} \E\big[\sum_{w\in \mathcal T}g(\zeta_{w}){\bf 1}_{\{d_{w} < T-s\}}\big]_{|s = d_u} {\bf 1}_{\{ d_u < T \}}  \Big] 
\end{align*}
conditioning with respect to $\mathcal F_{d_u}$ on $\{d_u < T\}$ and applying the branching property.
Next, we have
$$\E\big[\sum_{w\in \mathcal T}g(\zeta_{w}){\bf 1}_{\{d_{w} < T-s\}}\big] = \E\big[\sum_{w\in \mathring{\mathcal T}_{T-s}}g(\zeta_{w})\big]=\frac{1}{m}\int_0^{T-s}e^{\lambda_Bt} P^t_{H_B}\big(gH_B\big)(0)dt$$
by \eqref{Mt1 int} of Proposition~\ref{prop: Mt1}. Since $\{ ui \in \mathcal T \}  = \{ i \leq \nu_u \}$, and $\nu_u$ is independent of $\zeta_u$ and $d_u$ and has distribution with expectation $m$, we obtain
\begin{align*}
\E\big[\sum_{u,v \in \mathring{\mathcal T}_T, \atop u \prec v}g(\zeta_u)g(\zeta_v)\big] & = \E\Big[\sum_{u \in \mathring{{\mathcal T}}_T} g(\zeta_u)  \int_0^{T-d_u}e^{\lambda_Bt} P^t_{H_B}\big(gH_B\big)(0)dt \Big] 
\end{align*}
and we conclude by using once more \eqref{Mt1 int} of Proposition~\ref{prop: Mt1} (slightly generalized for test functions that depend on $d_u$ and $\zeta_u$).
\subsection{Proof of Lemma~\ref{lem: renouvellement}}
Let $\tau$ denote the first jump time of the process $(\widetilde \chi_t)_{t \geq 0}$.  Conditioning on $\{\tau >t\}$ and applying the strong Markov property yields
\begin{align*}
P^t_{H_B}\big(gH_B\big)(0) &
 = g(t)H_B(t) \PP(\tau > t) + \int_0^t P_{H_B}^{t-u}\big(gH_B\big)(0)f_{H_B}(u)du.
\end{align*}
The function $t \leadsto u(t)=P^t_{H_B}\big(gH_B\big)(0)$ satisfies a renewal equation of the form $u=u_0+ u \star f_{H_B}$, with locally bounded initial condition $u_0=gH_B\PP(\tau > \cdot)$ and renewal distribution $f_{H_B}(y)dy$. Its unique solution is given by 
$$P^t_{H_B}\big(gH_B\big)(0)  = g(t)H_B(t) \PP(\tau > t)+\int_0^t g(t-s)H_B(t-s) \PP(\tau > t-s)d\E[\widetilde N_s],$$
where $\widetilde N_t = \sum_{s \leq t} {\bf 1}_{\{\widetilde \chi_s-\widetilde \chi_{s^-}<0\}}$ is the counting process associated to $(\widetilde \chi_t)_{t \geq 0}$. By construction, we have $\E[\widetilde N_t]=\E\big[\int_0^t H_B(\widetilde \chi_{s})ds\big]$ and $\PP(\tau>t) = \int_t^\infty f_{H_B}(y)dy = m\int_t^\infty e^{-\lambda_By}f_B(y)dy \leq me^{-\lambda_Bt}$, therefore 
$$|P^t_{H_B}\big(gH_B\big)(0)| \leq |g(t)|e^{-\lambda_B t} m|H_B|_\infty+ |H_B|_\infty^2\int_0^t |g(u)|du$$
and we obtain the desired estimate thanks to the fact that $H_B$ is uniformly bounded over $\mathcal B$. 

\vip

\noindent {\bf Acknowledgements.} We are grateful to V. Bansaye. M. Doumic for helpful discussion and comments. The illuminating coupling argument for proving Proposition~\ref{prop: convergence semigroupe} was indicated to us by N. Fournier. The suggestions of two referees helped to considerably improve a former version of this work. Part of this work was completed while M.H. was visiting Humboldt-Universit\" at zu Berlin. The research of M.H. is partly supported by the Agence Nationale de la Recherche, (Blanc SIMI 1 2011 project {\small CALIBRATION}).

\end{document}